\documentclass[12pt]{amsart}
\usepackage{txfonts}
\usepackage{amssymb}
\usepackage[mathscr]{eucal}
\usepackage{amsmath}
\usepackage{amscd}
\usepackage[dvips]{color}
\usepackage{multicol}
\usepackage[all]{xy}
\usepackage{graphicx}
\usepackage{color}
\usepackage{colordvi}
\usepackage{xspace}
\usepackage[colorlinks,final,backref=page,hyperindex,hypertex]{hyperref}

\usepackage{tikz}
\usepackage[active]{srcltx}

\begin{document}

\newcommand{\nc}{\newcommand}
\newcommand{\delete}[1]{}

\nc{\mlabel}[1]{\label{#1}}  
\nc{\mcite}[1]{\cite{#1}}  
\nc{\mref}[1]{\ref{#1}}  
\nc{\mbibitem}[1]{\bibitem{#1}} 

\delete{
\nc{\mlabel}[1]{\label{#1}  
{\hfill \hspace{1cm}{\bf{{\ }\hfill(#1)}}}}
\nc{\mcite}[1]{\cite{#1}{{\bf{{\ }(#1)}}}}  
\nc{\mref}[1]{\ref{#1}{{\bf{{\ }(#1)}}}}  
\nc{\mbibitem}[1]{\bibitem[\bf #1]{#1}} 
}

\newtheorem{theorem}{Theorem}[section]
\newtheorem{thm}[theorem]{Theorem}
\newtheorem{prop}[theorem]{Proposition}
\newtheorem{lemma}[theorem]{Lemma}
\newtheorem{coro}[theorem]{Corollary}
\newtheorem{cor}[theorem]{Corollary}
\newtheorem{prop-def}[theorem]{Proposition-Definition}
\newtheorem{claim}{Claim}[section]
\newtheorem{propprop}{Proposed Proposition}[section]
\newtheorem{conjecture}[theorem]{Conjecture}
\newtheorem{assumption}{Assumption}
\newtheorem{condition}[theorem]{Assumption}
\newtheorem{question}[theorem]{Question}
\theoremstyle{definition}
\newtheorem{defn}[theorem]{Definition}
\newtheorem{exam}[theorem]{Example}
\newtheorem{remark}[theorem]{Remark}
\newtheorem{ex}[theorem]{Example}
\newtheorem{conv}[theorem]{Convention}

\renewcommand{\labelenumi}{{\rm(\alph{enumi})}}
\renewcommand{\theenumi}{\alph{enumi}}
\renewcommand{\labelenumii}{{\rm(\roman{enumii})}}
\renewcommand{\theenumii}{\roman{enumii}}
\nc {\SSubP}{discrete open subdivision property\xspace}
\nc {\ISubP}{continuous subdivision property\xspace}
\nc {\ValP}{discrete closed subdivision property\xspace}

\nc {\conefamilyc}{\underline{{C}}}
\nc {\conefamilyd}{\underline{{D}}}
\nc {\conefamilye}{\underline{{E}}}

\newcommand{\frakA}{\mathfrak{A}}
\newcommand{\frakC}{\mathfrak{C}}
\newcommand{\frakU}{\mathfrak{U}}
\nc{\Ima}{\operatorname{Im}}            
\nc{\Dom}{\operatorname{Dom}}      
\nc{\Diff}{\operatorname{Diff}}           
\nc{\End}{\operatorname{End}}        
\nc{\Id}{\operatorname{Id}}                 
\nc{\Isom}{\operatorname{Isom}}      
\nc{\Ker}{\operatorname{Ker}}           
\nc{\Lin}{\operatorname{Lin}}             
\nc{\Res}{\operatorname{Res}}         
\nc{\spec}{\operatorname{sp}}           
\nc{\supp}{\operatorname{supp}}      
\nc{\Tr}{\operatorname{Tr}}                 
\nc{\Vol}{\operatorname{Vol}}            
\nc{\sign}{\operatorname{sign}}         
\nc{\id}{\operatorname{id}}
\nc{\lin}{\operatorname{lin}}
\nc{\I}{J}
\nc{\cala}{\mathcal{A}}

\nc{\ot}{\otimes}
\nc{\bfk}{\mathbf{k}}
\nc{\wvec}[2]{{\scriptsize{\big [ \!\!
    \begin{array}{c} #1 \\ #2 \end{array} \!\! \big ]}}}

\nc{\zb}[1]{\textcolor{blue}{ #1}}
\nc{\li}[1]{\textcolor{red}{ #1}}
\nc{\sy}[1]{\textcolor{purple}{  #1}}

\newcommand{\Z}{\mathbb{Z}}
\newcommand{\ZZ}{\mathbb{Z}}
\newcommand{\Q}{\mathbb{Q}}               
\newcommand{\QQ}{\mathbb{Q}}               
\newcommand{\R}{\mathbb{R}}               
\newcommand{\RR}{\mathbb{R}}               
\newcommand{\coof}{L}           
\newcommand{\coef}{\QQ}
\newcommand{\p}{\partial}         
\nc{\cone}[1]{\langle #1\rangle}
\nc{\cl}{c}                 
\nc{\op}{o}                 
\nc{\ccone}[1]{\langle #1\rangle^\cl}
\nc{\ocone}[1]{\langle #1\rangle^\op}
\nc{\cc}{\mathfrak{C}}      
\nc{\dcc}{\mathfrak{DC}}    
\nc{\oc}{\mathcal{C}^o}      
\nc{\dsmc}{\dcc }  
\nc{\csup}{{^\ast}}
\nc{\bs}{\check{S}\,}
\nc{\ci}{{C,0}}
\nc{\cii}{{C,1}}
\nc{\ciii}{{C,\geq 2}}
\nc{\civ}{\mathrm{SI}}
\nc{\cv}{{N}}

 \nc {\linf}{{\rm lin} (F)^\perp}
\nc{\dirlim}{\displaystyle{\lim_{\longrightarrow}}\,}
\nc{\coalg}{\mathbf{C}}
\nc{\barot}{{\otimes}}

\newcommand{\one}{\mbox{$1 \hspace{-1.0mm} {\bf l}$}}
\newcommand{\A}{\mathcal{A}}              
\newcommand{\Abb}{\mathbb{A}}          
\renewcommand{\a}{\alpha}                    
\renewcommand{\b}{\beta}                       

\newcommand{\B}{\mathcal{B}}              
\newcommand{\C}{\mathbb{C}}
 \newcommand{\calm}{{\mathcal M}}

\newcommand{\CC}{\mathcal{C}}           
\newcommand{\CR}{\mathcal{R}}           
\newcommand{\D}{\mathbb{D}}               
\newcommand{\del}{\partial}                    
\newcommand{\DD}{\mathcal{D}}           
\newcommand{\Dslash}{{D\mkern-11.5mu/\,}} 
\newcommand{\e}{\varepsilon}            
\newcommand{\F}{\mathcal{F}}                
\newcommand{\Ga}{\Gamma}                  
\newcommand{\ga}{\gamma}                   
\renewcommand{\H}{\mathcal{H}}           
\newcommand{\half}{{\mathchoice{\thalf}{\thalf}{\shalf}{\shalf}}}
\newcommand{\hideqed}{\renewcommand{\qed}{}} 
\newcommand{\K}{\mathcal{K}}             
\renewcommand{\L}{\mathcal{L}}          
\newcommand{\la}{\lambda}                   
\newcommand{\<}{\langle}
\renewcommand{\>}{\rangle}
\newcommand{\M}{\mathcal{M}}            
\newcommand{\Mop}{\star}                     
\newcommand{\N}{\mathbb{N}}             
\newcommand{\norm}[1]{\left\lVert#1\right\rVert}    
\newcommand{\norminf}[1]{\left\lVert#1\right\rVert_\infty} 
\newcommand{\om}{\omega}                 
\newcommand{\Om}{\Omega}                
\newcommand{\ol}{\\widetilde}                  
\newcommand{\OO}{\mathcal{O}}          
\newcommand{\ovc}[1]{\overset{\circ}{#1}}
\newcommand{\ox}{\otimes}                    
\newcommand{\pa}{\partial}
\newcommand{\piso}[1]{\lfloor#1\rfloor} 

\newcommand{\rad}{{\mathbf r}}
\newcommand{\sepword}[1]{\quad\mbox{#1}\quad} 
\newcommand{\set}[1]{\{\,#1\,\}}               
\newcommand{\shalf}{{\scriptstyle\frac{1}{2}}} 
\newcommand{\slim}{\mathop{\mathrm{s\mbox{-}lim}}} 
\renewcommand{\SS}{\mathcal{S}}        
\newcommand{\Sp}{{\rm Sp}}
\newcommand{\sg}{\sigma}                              
\newcommand{\T}{\mathbb{T}}                
\newcommand{\tG}{\widetilde{G}}           
\newcommand{\thalf}{\tfrac{1}{2}}            
\newcommand{\Th}{\Theta}
\renewcommand{\th}{\theta}
\newcommand{\tri}{\Delta}                        
\newcommand{\Trw}{\Tr_\omega}           
\newcommand{\UU}{\mathcal{U}}              
\newcommand{\Afr}{\mathfrak{A}}           
\newcommand{\vf}{\varphi}                       
\newcommand{\x}{\times}                          
\newcommand{\wh}{\widehat}                  
\newcommand{\wt}{\widetilde}                 
\newcommand{\ul}[1]{\underline{#1}}             
\renewcommand{\.}{\cdot}                          
\renewcommand{\:}{\colon}                       
\newcommand{\comment}[1]{\textsf{#1}}

\nc{\calc}{\mathcal{C}}
\nc{\calf}{\mathcal{F}(C\sim \cup _{i=1}^nC_i)}
\nc {\cals}{\mathcal {S}}
\nc{\calh}{\mathcal{H}}
\nc{\deff}{K}
\nc{\cali}{\mathcal{I}}
\nc{\calp}{\mathcal{P}}
\nc{\calq}{\mathcal{Q}}
\nc{\calt}{\mathcal{T}}
\nc{\vep}{\varepsilon}
\nc {\ltcone}{lattice cone\xspace}
\nc{\ltcones}{lattice cones\xspace}
\nc{\abf}{Algebraic Birkhoff Factorization\xspace}
\nc{\abfs}{Algebraic Birkhoff Factorizations\xspace}
\nc {\lC}{(C, \Lambda _C)}
\nc {\rdim}{{\rm dim}}

\nc{\jh}{H}
\nc{\jf}{F}

\title[Birkhoff Factorization and Euler-Maclaurin formula]{Algebraic Birkhoff Factorization and the Euler-Maclaurin formula on cones}

\author{Li Guo}
\address{Department of Mathematics and Computer Science,
         Rutgers University,
         Newark, NJ 07102, USA}
\email{liguo@rutgers.edu}

\author{Sylvie Paycha}
\address{Institute of Mathematics,
University of Potsdam,
Am Neuen Palais 10,
D-14469 Potsdam, Germany}
\email{paycha@math.uni-potsdam.de}

\author{Bin Zhang}
\address{School of Mathematics, Yangtze Center of Mathematics,
Sichuan University, Chengdu, 610064, P. R. China}
\email{zhangbin@scu.edu.cn}

\date{\today}

\begin{abstract} We equip the space of lattice cones with a  coproduct which makes it a connected cograded colagebra. The exponential sum and exponential integral on lattice cones can be viewed as linear maps on this space with values in the space of meromorphic germs with linear poles at zero. We investigate the subdivision properties-- reminiscent of the inclusion-exclusion principle for the cardinal on finite sets-- of such linear maps     and establish a compatibility of these properties with respect to the convolution quotient of the coalgebra.  Implementing the \abf procedure on the  linear maps under consideration, we factorize the exponential sum as a convolution quotient of two maps, with each of the maps in the factorization satisfying a subdivision property. Consequently, the \abf specializes to the Euler-Maclaurin formula on lattice cones and provides a simple formula for the interpolating factor by means of a projection map.
\end{abstract}

\subjclass[2010]{11H06, 52C07, 52B20, 65B15, 11M32}

\keywords{convex cones, coalgebras, \abf, Euler-Maclaurin formula, meromorphic functions, subdivision property }

\maketitle
\vspace{-1.3cm}

\tableofcontents

\setcounter{section}{0}

\allowdisplaybreaks

\section{Introduction}

The classical Euler-Maclaurin formula in analysis \cite{Ha} and its higher dimensional generalizations~\cite{BrV,CS,PK} express Riemann sums in terms of integrals over polytopes.
Their geometric relevance in relation with
the Riemann-Roch theorem on toric varieties arises from the appearance of the Todd operators \cite{GS},  related to the Todd classes of the toric varieties associated with the polytopes.
We study these formulae applying \abf from a renormalization method in quantum field theory.

The idea of this approach comes from two observations. In~\cite {BZ} localized formulae for equivariant Todd classes of toric varieties are given,  which make explicit the geometric nature of the localized formulae. A natural question is how to recover the equivariant or ordinary Todd classes from the localized formulae.  From a  mathematical viewpoint, this amounts to extracting them from fractions arising in the localized formula, and from  the viewpoint of physics, it boils down    to dealing with the singularities, an issue which calls for a  renormalization procedure. On the other hand, the exponential sum on a cone can be viewed as a regularization of the ill-defined partition function $\sum\limits_{\vec n\in C\cap \Z^k} 1$ over a cone $C$ in $\R^k$. This suggests the application of a renormalization process.

Precisely, on a convex polyhedral convex cone, the exponential sum in Eq.~(\ref{eq:rcczv})  and exponential integral in Eq.~(\ref{eq:ExpI}) can be viewed as morphisms with values in the space of multivariate meromorphic germs with linear poles at zero.
We interpolate the exponential sum and exponential integral
by means of an \abf  implemented on geometric cones, inspired by the algebraic renormalization scheme of Connes and Kreimer. Let us briefly recall their approach.
\smallskip

\noindent
{\bf Theorem} ({\bf \abf})\cite{CK}
{\it Let $H$ be a commutative connected filtered Hopf algebra.
Let $R$ be a commutative algebra with a Rota-Baxter operator $P$ of weight $-1$.
Let $\phi: H \to R$ be an algebra homomorphism.
\begin{enumerate}
\item
There are algebra homomorphisms $\phi_-: H \to \bfk+P(R)$ and
$\phi_+: H \to \bfk+(\id-P)(R)$, with $\bfk$ being the base ring,  such that
$$\phi=\phi_-^{\ast\, (-1)}\ast \phi_+.$$
Here $\phi_-^{\ast\, (-1)}$ is the inverse of $\phi_-$ with respect to the convolution product $\ast$ on the space of linear maps from $H$ to $R$ associated with the coproduct on $H$.
\mlabel{it:decom}
\item
If $P^2=P$, then the decomposition in~(\mref{it:decom}) is unique.
\mlabel{it:uni}
\end{enumerate}}

In our context, the projection $P$ does not satisfy  the Rota-Baxter property, so we first need to generalize Connes and Kreimer's approach.  By identifying the factors in the \abf,   we then show how the \abf indeed gives the Euler-Maclaurin formula. This approach has the extra benefit of providing a simple formula for the interpolation function.
\smallskip

As the context to apply the \abf, we introduce the notion
of lattice cones (Definition \ref{defn:latticecone}), which are pairs consisting of a cone and a lattice, needed to make sense of exponential generating sums relative to a choice of lattice points. On lattice cones, the exponential generating sum
  $S^c$ in Eq.~(\ref{eq:rcczv})  and exponential integral  $I$  in Eq.~(\ref{eq:ExpI}), first defined on simplicial lattice cones and then extended to general lattice cones by subdivisions, yield meromorphic germs with linear poles. Thus, the linear extensions to the linear space generated by lattice cones give linear maps with values  in the space of meromorphic germs with linear poles at zero.

  To construct the coproduct in the space of lattice cones needed for implementing the \abf, we fix an inner product (see Eq.~(\ref{eq:Q})) on the underlying  space of the lattice cones. Borrowing the definition of transverse cone from~\cite{BV1}, defined by means of   this inner product, we build the   coproduct in Eq.~(\ref{eq:coproduct}) on the space of lattice cones from a complement map which assigns to a face of a lattice cone  the transverse \ltcone (Proposition~\ref{pp:transversecone}).  This coproduct    is compatible with the partial order and the dimension filtration on cones; Theorem \ref{thm:Hopfoncones} endows  the space of lattice cones   with a connected cograded  coalgebra structure.
The corresponding convolution product  (Lemma~\ref{lem:phiinv}) on the algebra of linear maps from the space of lattice cones to a commutative algebra is   later used for the \abf.

  The \abf in the renormalization scheme of Connes and Kreimer requires the regularized linear map to take values in a Rota-Baxter algebra. The fact that the range of our linear maps being the space of meromorphic germs with linear poles at zero imposes a special treatment. This   is one of our motivations   to investigate the structure of the space of meromorphic germs with linear poles at zero~\cite {GPZ3}. It turns out that this space is a commutative algebra which splits into a subalgebra and a complement of it which is not a subalgebra. Consequently, the projection is not a Rota-Baxter operator, and the decomposition depends on a choice of an inner product. Theorem~\ref{thm:abf} which in contrast does apply to  the present situation,  generalizes  the \abf to  linear maps on a connected cograded coalgebra which is not necessarily a Hopf algebra, with values in a commutative algebra which splits into a subalgebra and its  complement.

Having the necessary ingredients at hand, we then apply the (generalized) \abf to the exponential generating sum, and obtain a factorization in terms of a ``holomorphic" part and a ``polar part" (Corollary \ref{cor:abfd}).

 Our next step is to derive the Euler-Maclaurin formula as a special case of the \abf, when the inner product used to define the transverse cone is assumed to coincide with the inner product to define the projection $\pi_+$ (in Eq.~(\ref{eq:pi+})) onto the holomorphic part of the space of meromorphic germs with linear poles at zero. For this purpose, we only need to identity the ``polar part" of the \abf with the exponential integral in the Euler-Maclaurin formula, which is clear for smooth cones. In order to apply it to general lattice cones by means of subdivisions, we carry out a detailed study of the different types of subdivision properties (Definition~\ref{defn:ValP}) enjoyed by the exponential generating sum and exponential integral, including closed discrete type for the sum  and  of continuous type for the integral. This is reminiscent of the inclusion-exclusion principle in set theory and the sieve method in number theory.

The compatibility of subdivision properties of the factors with the convolution quotient in the \abf is investigated in the general result Theorem~\ref{thm:IsubP}. It states that  the convolution quotient of two maps on the coalgebra of lattice cones with values in a commutative algebra, both of which satisfy the \ValP, satisfies the \ISubP.

Returning to our case of the exponential generating sum and exponential integral, the fact (Theorem \ref{prop:mu}) that the  ``holomorphic part" coincides with the holomorphic projection of the exponential generating sum, implies that it satisfies the \ValP (see Corollary~\ref{co:subd}.(\ref{cor:musubd})).  Theorem \ref{thm:IsubP} applied to the \abf of the exponential generating sum, leads to Corollary~\ref{co:subd}.(\ref{thm:MuSub}) which states that the ``polar part" satisfies the continuous subdivision property.
Based on the fact which results from a straightforward calculation,  that for smooth \ltcones , the ``polar part" is the exponential integral, the compatibility with subdivisions yields that the ``polar part"  coincides with the exponential integral for general \ltcones. Consequently, the \abf amounts to the  the Euler-Maclaurin  formula.

\section{Lattice cones and their coproduct}
\mlabel{sec:coalg}
In this section, we introduce the concepts of lattice cones and a transverse lattice cones to faces of lattice cones. Using transverse lattice cones, we equip the linear span of lattice cones with a coalgebra structure.

\subsection{Lattice cones} In a finite dimensional vector space over $\R$, a {\bf lattice} is a finitely generated subgroup which spans the whole space. A real vector space  equipped  with a lattice is called a {\bf lattice space} A rational multiple of a vector in the lattice is called a {\bf rational lattice vector}.

\begin{defn}
A {\bf filtered lattice space} is a pair $(V,\Lambda)$ from a family $(V_k,\Lambda_k), k\geq 1,$ of lattice vector spaces such that $V_{1}\subset V_2\subset \cdots $, $V=\cup_{k=1}^\infty V_k$, $\Lambda _k=\Lambda _{k+1}\cap V_k$ and $\Lambda=\cup_{k=1}^\infty \Lambda _k$.
 \end{defn}
\begin{remark}In applications, the filtered lattice space usually is $\RR ^\infty$ with $V_k=\RR ^k$,  $\Lambda_k$ the standard lattice $\ZZ^k$, and $\{e_1, e_2, \cdots \}$ the canonical basis.
 \end{remark}

We now collect basic definitions and facts (mostly following~\mcite{Fu} and \mcite{Zi}) on cones that will be used in this paper. See~\cite {GPZ2} for a detailed discussion on these facts. For a subset $S$ of $V$, let $\lin(S)$ denote its $\R$-linear span.

\begin{enumerate}
\item
By a {\bf cone} in $V_k$ we mean a {\bf closed convex polyhedral cone} in $V_k$, namely the convex set
\begin{equation}
\cone{v_1,\cdots,v_n}:=\RR\{v_1,\cdots,v_n\}=\RR _{\geq 0}v_1+\cdots+\RR_{\geq 0}v_n,
\mlabel{eq:cone}
\end{equation}
where $v_i\in \Lambda_k$, $i=1,\cdots, n$.
\item
The set $\{v_1,\cdots,v_n\}$ in Eq.~(\mref{eq:cone}) is called a {\bf generating set} or a {\bf spanning set} of the cone.
\item
The spanning set $\{v_1,\cdots,v_n\}$ is called {\bf primary} if
\begin{enumerate}
\item $v_i\in \Lambda_k$, $i=1, \cdots , n$,
\item there is no real number $r_i\in (0,1)$ such that $r_iv_i$ lies in $ \Lambda_k$, and
\item none of the generating vectors $v_i$ is a positive linear combination of the others.
\end{enumerate}
For a lattice cone, its primary generating set exists.
\item Define the {\bf dimension} of  a cone $C$ by $\dim C:=\dim \lin (C)$.
\item
A cone is called {\bf strongly convex} if it does not contain any nonzero linear subspace.
\item
A {\bf simplicial cone} is a cone spanned by linearly independent vectors. A simplicial cone is strongly convex.
\item
A {\bf smooth cone} is a cone whose primary generating set is a part of a lattice basis of $\Lambda _k\subseteq V_k$. For a full dimensional cone, smoothness is equivalent to the {\bf unimodularity}, namely that the determinant of the transformation matrix relating the primary generating set  to a basis of $\Lambda_k$ is $\pm 1$.
\item A {\bf face} of a cone $C$ is a subset of the form $C\cap \{u=0\}$, where $u:V_k\to \RR$
is a linear form that is non-negative on $C$. A face $F$ of a cone $C$ is again a cone and we write $F\preceq C$. If $F$ is a proper face  of a cone $C$ we write $F\precneqq C$.
\end{enumerate}

\begin{ex}A Chen cone $C_k^{\mathrm{Chen}}$, defined by $\langle e_1,e_1+e_2,\cdots, e_1+\cdots+e_k\rangle$, is a smooth cone.
\end{ex}

\begin {lemma}
\mlabel {lem:rational}
Let $W\subset U$ be subspaces of lattice space $(V_k, \Lambda _k)$ spanned by lattice vectors and let $\Lambda_U$ be a lattice of $U$ with lattice vectors. Then $W\cap \Lambda _U$ is a lattice of $W$.
\mlabel{lem:latt}
\end {lemma}

\begin{proof} Let $\{w_1, \cdots, w_m\}$ be a basis of $W$ with lattice vectors and $\{u_1, \cdots , u_\ell\}$ a basis of $\Lambda _U$ with lattice vectors. Then for $i=1, \cdots m$, $w_i$ is a rational combination of $u_1, \cdots , u_\ell$. Therefore there exist $0\neq r_i\in \ZZ $ such that $r_iw_i \in \Lambda _U, i=1,\cdots, m$. Then we have $W=\sum\limits_{i=1}^m \R r_i w_i \subseteq \RR (W\cap \Lambda_U)$. Since $W\cap \Lambda _U$ is also finitely generated, it is a lattice of $W$.
\end{proof}

On the grounds of Lemma~\mref{lem:latt}, we set the following definition.
\begin{defn}
\begin{enumerate}
\item
A {\bf \ltcone}\footnote{The relevance of a chosen lattice in a vector space is mentioned in \cite{BV1} (see the word of caution in par. 4). The term lattice cone can also be found in the literature on Banach spaces with a somewhat different meaning~\cite{Lo}.}  in $V_k$ is a pair  $(C, \Lambda _C)$  with $C$  a cone in $V_k$ and $\Lambda _C$  a lattice in ${\rm lin}(C)$  generated by lattice vectors.
\item A {\bf face} of a \ltcone $(C,\Lambda_C)$ is the \ltcone \ $(F, \Lambda _F)$ where $F$ is a face of $C$ and $\Lambda _F:=\Lambda _C\cap {\rm lin}(F)$.
\item
A {\bf primary generating set} of a {\ltcone} \ $\lC$ is a generating set $\{v_1,\cdots,v_n\}$ of $C$ such that
\begin{enumerate}
\item
$v_i\in \Lambda_C$, $i=1, \cdots , n$,
\item
there is no real number $r_i\in (0,1)$ such that $r_iv_i$ lies in $ \Lambda_C$, and
\item none of the generating vectors $v_i$ is a positive linear combination of the others.
\end{enumerate}
\end{enumerate}
\end{defn}

\begin{remark}
\begin{enumerate}
\item Any {\ltcone} possesses a primary generating set: starting from any lattice generating set $\{v_1, \cdots, v_n\}$ of the \ltcone, a rescaling yields a set satisfying the first two conditions.
A primary generating set is obtained by eliminating an element if it is a combination of the remaining ones.
\item  For a  strongly convex {\ltcone}, a primary generating set is unique: it consists of the shortest lattice vector in each of the spanning vectors of cone~\cite {Fu}.
\item For a cone $C\subseteq V_k$, the primary generating set of the lattice cone $(C,\lin (C)\cap \Lambda_k)$ coincides with that of the cone $C$.
\end{enumerate}
\end{remark}

The following properties of \ltcones are easy to verify.
\begin{lemma}
Let $(C,\Lambda_C)$ be a lattice cone.
\begin{enumerate}
\item
Let $C'$ be a lattice cone with $\lin(C)=\lin(C')$. Then $(C',\Lambda_C)$ is also a lattice cone.
\mlabel{it:latt1}
\item
If $F$ is the face of another face $G$ of $C$, then $\Lambda_F=\Lambda_G\cap \lin(F)$.
\end{enumerate}
\mlabel{rk:latt}
\end{lemma}

\begin{ex} The lattice cone $\left( \langle   e_1\rangle,\Z e_1\right)$ is a face of the lattice cone  $\left(\langle e_1,  e_2\rangle,\Z e_1+ \Z e_2\right)$. It is also a face of the lattice cone    $ \left(\langle e_1, e_2\rangle, \Z( e_1+e_2)+ \Z e_2\right)$ since $ r e_1\in  \Z( e_1+e_2)+ \Z e_2$ if and only if $r\in \Z$. But it is not a face of the lattice cone $ \left(\langle e_1, e_2\rangle, \Z( e_1+e_2)+ \Z (e_1-e_2)\right)$ since
 $ r e_1\in  \Z( e_1+e_2)+ \Z (e_1-e_2 )$ if and only if $r\in 2\Z.$ This also shows that $(\langle e_1\rangle, \Z 2e_1)$ is a face of the lattice cone $ \left(\langle e_1, e_2\rangle, \Z( e_1+e_2)+ \Z (e_1-e_2)\right)$.
 \end{ex}

\begin{defn}\label{defn:latticecone}
A \ltcone  \ $(C, \Lambda _C)$  is called {\bf strongly convex} (resp. {\bf simplicial}) if $C$ is. A \ltcone  \ $(C, \Lambda _C)$ is called {\bf smooth} if the additive monoid $\Lambda_C\cap C$ has a monoid basis. In other words, $(C, \Lambda _C)$ is called {\bf smooth} if and only if there are linearly independent lattice vectors $v_1,\cdots,v_\ell$ such that
$\Lambda_C\cap C=\ZZ_{\geq 0}\{v_1,\cdots,v_\ell\}$.
\end{defn}

The following facts are easy to check.
\begin{remark}
\begin{enumerate}
\item For any simplicial cone $C$ spanned by linearly independent lattice vectors $v_1,\cdots,v_n$, the lattice cone $(C,\ZZ\{v_1,\cdots,v_n\})$ is smooth;
\item The smoothness of a cone comapre with that of a lattice cone, for a cone $C$ in $V_k$ is smooth if and only if the lattice cone $(C,\Lambda_k\cap {\rm lin}(C))$ is smooth.
\end{enumerate}
\end{remark}

\begin{ex} The lattice cone $(\cone{e_1,e_2}, \ZZ e_1+\ZZ e_2)$ is smooth. By the first remark, the lattice cone $(\cone{e_1,e_1+2e_2}, \ZZ e_1 + \ZZ 2e_2)$ is smooth even though $\cone{e_1,e_1+2e_2}$ is not smooth. By the second remark the lattice cone $(\cone{e_1,e_1+2e_2},\ZZ e_1+\ZZ e_2)$ is not smooth.
\end{ex}

The following elementary property is useful for later purposes.

\begin{prop}
A face of a smooth \ltcone  \ $(C, \Lambda _C)$  is smooth.
\mlabel{pp:smface}
\end{prop}

\begin {proof} Let  $(F, \Lambda _F)$ be a face of a smooth \ltcone \ $\lC$. Let  $u:V_k\to \RR$  be  a linear function  defining the face $F\colon=C\cap u^\perp$. Then $\Lambda_F=\Lambda_C\cap \lin(F)$.  Let $\{v_1, \cdots , v_m\}$ be a monoid basis of $C\cap \Lambda_C$. To prove that $(F,\Lambda_F)$ is smooth, we only need to show that the set $\{v_1,\cdots,v_m\}\cap u^\perp$ is a monoid basis of $\Lambda_F \cap F$. Since the set is linearly independent, this amounts to  showing that it generates $\Lambda_F\cap F$ as a monoid.

Let $v\in \Lambda_F\cap F$. Note that $\Lambda_F\cap F=\Lambda_C\cap \lin(F)\cap C\cap u^\perp = \Lambda_C\cap C\cap u^\perp$. Thus for $v\in \Lambda_F\cap F$, from $v\in \Lambda_C$ we have
$v=\sum\limits_{i=1}^m a_i v_i$ with $a_i\in \ZZ$. From $v\in C$ we also have $a_i\in \ZZ_{\geq 0}$. For $v\in u^\perp$ we further have
$0=u(v)=\sum\limits_{i=1}^m a_i u(v_i)$. Thus if $a_i\neq 0$, then $u(v_i)=0$. It follows that
$v\in \sum\limits_{i, u(v_i)=0} \ZZ_{\geq 0} v_i = \ZZ_{\geq 0}\{\{v_1,\cdots,v_m\}\cap u^\perp\}$, which completes the proof.
\end{proof}

\subsection{Transverse {\ltcone}s}

Let $\cc_k$ denote the set of {\ltcone}s in $V_k$, $k\geq 1$.
The natural inclusions $\cc_k\to \cc_{k+1}$ induced by the natural inclusions $V_k \to V_{k+1}$, $\Lambda _k\to \Lambda _{k+1}, \ k\geq 1,$
give rise to the direct  limit $\cc =\dirlim \cc_k= \cup_{k\geq 1} \cc_k$.

We want to equip the $\QQ$-linear space $\QQ \cc$ generated by $\cc$ with a coproduct by applying the concept of a transverse cone  borrowed from~\mcite{BV1} and enriched to \ltcones.

We use an inner product on a filtered lattice space to identify quotient spaces and subspaces. This can be done by means of more general complement maps as in \cite {GP} but, in this paper, we choose to use  the inner product for that purpose.

\begin{defn} \label{conv:scalarproduct}
Let $V:=\cup_{k\geq 1}V_k$ with $\Lambda=\cup_{k\geq 1}\Lambda_k$ be a filtered lattice space. An {\bf inner product} $Q(\cdot,\cdot)=(\cdot,\cdot)$ on $V$ is a sequence of inner products
$$ Q_k(\cdot,\cdot)=(\cdot,\cdot)_k: V_k\ot V_k \to \RR, \quad k\geq 1,$$
that is compatible with the inclusion $j_k:V_k\hookrightarrow V_{k+1}$ and whose restriction to $\Lambda\ot \QQ$ and hence $\Lambda$ takes values in $\QQ$. A filtered lattice vector space together with an inner product on $V$ is called a \bf{filtered lattice Euclidean space}.
\end{defn}

From now on, our discussion is on a fixed filtered lattice Euclidean space $\left(V,\Lambda\right)$ with the Euclidean inner product  \begin{equation}\label{eq:Q}Q(\cdot, \cdot)=(\cdot, \cdot)\end{equation} and we drop $Q$ from the superscript to simplify notations whenever there is no ambiguity.
Let $L$ be a lattice subspace of $V_k$.
Set
$$L^{\perp_k}:=L^{\perp_k^Q}:=\left \{ v\in V_k\,|\, Q_k(v,u)  =0\text{ for all } u\in L\right\}.$$
The inner product $Q_k$ gives the direct sum decomposition $V_k=L\oplus L^{\perp_k}$ and hence the orthogonal projection
\begin{equation} \pi_{k,L^\perp}:=\pi_{k,L^\perp}^Q: V_k \to L^{\perp_k}
\mlabel{eq:orthproj}
\end{equation}
along $L$.
Also, the induced isomorphism $Q_k^*:V_k\to V_k^*$ yields an embedding $V_k^*\hookrightarrow V_{k+1}^*$. We refer to the direct limit  $V^\circledast:=\bigcup_{k=1}^{\infty}V_{k}^*= \varinjlim V_{k}^*$   as the {\bf filtered dual space} of $V$. In general $ V^\circledast$ differs from the usual dual space $ V^*$.

\begin{ex} Let $V=\RR ^\infty $ be equipped with the canonical inner product. For $L=\lin( e_1+e_2)\subset V_2=\RR^2$. We have $V_2/L\simeq L^{\perp _2} =\lin(e_1-e_2)$.
 \end{ex}

\begin {defn}\label{defn:transversecone}  $($\cite {BV1}$)$
Let $F$ be a face of  a cone $C\subseteq V_k$. The {\bf transverse cone} $t(C,F)$ to $F$ is the projection
$\pi_{k,F^\perp}(C)$ of $C$ in $ \lin(F)^\perp\subseteq V_k$, where $\pi_{k,F^\perp}=\pi_{k,{\rm lin}(F)^\perp}$.
\end {defn}

Note that $t(C,F)$ might not be a face of $C$. For example, the  transverse cone to the face $F=\langle e_1+e_2\rangle$ of the cone $C=\langle e_1, e_1+e_2\rangle$ is the cone $t(C,F)=\langle e_1-e_2\rangle$ under the standard inner product.

The commutative diagram
\begin{equation}
  \mlabel{eq:traninc}\xymatrix{
V_k  \ar[rr]^{\pi_{k, F^{\perp}}} \ar@{_{(}->}[d]^{j_k} &&
F^{\perp_k} \ar@{^{(}->}[rr] \ar@{_{(}->}{[d]}^{j_k{ \vert_{F^{\perp_k}}}} && V_k \ar@{_{(}->}[d]^{j_k} \\
V_{k+1}  \ar[rr]^{\pi_{k+1, F^{\perp}}}  &&
F^{\perp_{k+1}} \ar@{^{(}->}[rr]&& V_{k+1}
}
\end{equation}
shows that $\pi_{k,F^\perp}(C)$ is actually independent of the choice of $k\geq 1$ such that $C\subseteq V_k$. Thus $t(C,F)$ is well-defined in $\cc$. So we can  simplify the above notations $\pi_{k, L^\perp}$ by dropping the subscript $k$.

\begin {lemma} For a face $F$ of $C$, the transverse cone $t(C,F)$ is strongly convex.
\mlabel {lem:StrongC}
\end{lemma}

\begin {proof} Assume that the face $F$ is given by a linear functional $u$, i.e., $F=C \cap \{u=0\}$. If the transverse cone $t(C,F)$ is not strongly convex, then there is a nonzero vector $v\in t(C,F)$, such that $-v\in t(C,F)$. By the definition of transverse cone, there are vectors $v'\in {\rm lin} (F)$ and $ v^{\prime \prime }\in {\rm lin} (F)$ such that $v+v'\in C$ and $-v+v^{\prime \prime}\in C$. Since $v$ is nonzero, we have $v+v'\not \in F$, so $u(v+v')=u(v)>0$. For the same reason, we have $u(-v)=u(-v+v^{\prime \prime})>0$, which is a contradiction.
\end{proof}

 We next generalize the concept of transverse cones to the context of \ltcones.  Let $(C,\Lambda_C)$ be a \ltcone in $V_k$. Under the projection $\pi_{F^\perp}:V_k\to \lin (F)^\perp$, the lattice cone $C$ is sent to a lattice cone. Also the lattice $\Lambda_C$ in $\lin(C)$ is sent to a lattice in $\pi_{F^\perp}(\lin(C))$ since $\pi_{F^\perp}(\Lambda_C)$ is a finitely generated abelian group and spans $\lin(t(C,F))= $ $\pi_{F^\perp}({\rm lin}(C))$.
This justifies the following definition.

\begin {defn} Let $(F, \Lambda_F)$ be a face of the \ltcone $(C, \Lambda_C)$ in $V_k$. The {\bf transverse \ltcone} $(t(C,F), \Lambda _{t(C,F)})$ along the face  $(F, \Lambda _F)$ is the image of $(C, \Lambda_C)$ under the projection $\pi_{F^\perp}$:
\begin{equation}
(t(C,F), \Lambda_{t(C,F)}):=(\pi_{F^\perp}(C), \pi_{F^\perp}(\Lambda_C)).
\mlabel{eq:tcdef}
\end{equation}
We also use the notation $t\left((C,\Lambda_C),(F,\Lambda_F)\right )$ to denote the transverse \ltcone.
\end {defn}
\begin{remark}  In general, $\Lambda _{t(C,F)}\not =\Lambda_C\cap {\rm lin}(t(C,F))$, see  the example below and  the word of caution in par. 4 of \cite{BV1}.
\end{remark}
\begin{ex}  Using the standard inner product and the induced lattice, the  transverse \ltcone \ to the face $(F, \Lambda _F)=(\langle e_1+e_2\rangle, \Lambda _2\cap {\rm lin}(e_1+e_2))$ of the cone $(C,\Lambda _C)=(\langle e_1, e_1+e_2\rangle, \Lambda _2)$ is $(t(C,F),\Lambda_{t(C,F)})=\left(\langle e_1-e_2\rangle, \ZZ \left(\frac {e_1-e_2}2\right)\right)$,
so that   $\Lambda _{t(C,F)}\not =\Lambda _2\cap {\rm lin}(e_1-e_2)=\ZZ(e_1-e_2)$.
\end{ex}

For faces $F\preceq G\preceq C$ of the cone $C$, the transverse cone $t(G,F)$ can be viewed as a face of $t(C,F)$ and as the transverse cone $\pi_{F^\perp}(G)$. Thus the lattice $\Lambda_{t(G,F)}$ of $t(G,F)$ can be defined in two ways, firstly as the lattice of the face $t(G,F)$ of $t(C,F)$, namely $\Lambda_{t(G,F)}:=\Lambda_{t(C,F)}\cap \lin(t(G,F))$, and alternatively as the lattice of the transverse cone $t(G,F)$, namely $\Lambda_{t(G,F)}:=\pi_{F^\perp}(\Lambda_G)$. We need to verify that the two definitions agree. For this we first prove a lemma.

\begin{lemma}We have
$\pi_{F^\perp}(\Lambda_C\cap \lin (G))
= \pi_{F^\perp}(\Lambda_C)\cap \pi_{F^\perp}(\lin(G)).$
\mlabel{lem:int}
\end{lemma}
\begin{proof}
The left hand side is clearly contained in the right hand side. On the other hand, for $w\in \pi_{F^\perp}(\Lambda_C)\cap\pi_{F^\perp}(\lin(G))$, there are $x\in \Lambda _C$ and $y\in \lin(G)$  such that $w=\pi _{F^\perp}(x)=\pi _{F^\perp}(y)$. So $\pi _{F^\perp}(x-y)=0$, that is, $x-y \in \lin (F)\subset \lin (G)$, implying $x \in \lin (G)$. Thus $x$ is in $\Lambda_C\cap \lin(G)$ and $w$ is contained in the left hand side.
\end{proof}

The following proposition shows the equivalence of the two definitions of  $\Lambda_{t(G,F)}$.\begin{prop}\label{prop:latticetrans}For $F\preceq G\preceq C$, we have $$\Lambda_{t(G,F)}:= \Lambda_{t(C,F)}\cap \lin(t(G,F))=\pi_{F^\perp}(\Lambda_G).$$
\end{prop}
\begin{proof}
Applying Lemma~\mref{lem:int}, we obtain
\begin{eqnarray*}
\pi_{F^\perp}(\Lambda_G)
&=& \pi_{F^\perp}(\Lambda_C\cap \lin (G))\\
&=& \pi_{F^\perp}(\Lambda_C)\cap \pi_{F^\perp}(\lin(G))\\
&=& \pi_{F^\perp}(\Lambda_C)\cap \lin (\pi_{F^\perp}(G))\\
&=& \Lambda_{t(C,F)}\cap \lin(t(G,F)).
\end{eqnarray*}
Thus the two definitions of $\Lambda_{t(G,F)}$ agree.
\end{proof}

\begin{prop}\mlabel{pp:transversecone}
Transverse cones enjoy the following properties. Let $F$ be a face of a cone $C$.
\begin{enumerate}
\item {\bf (Transitivity) } $t(C,F)=t\left( t(C,F^\prime), t(F, F^\prime)\right)$ if $F^\prime$ is a face of $F$.
    \mlabel{it:tra}
\item {\bf (Compatibility with the partial order) } We have
$ \{H\preceq t(C,F)\} = \{t(G,F)\,|\, F\preceq G\preceq C\}.$
\mlabel{it:com1}
    \item {\bf (Compatibility with the dimension filtration) }  ${\rm dim}(C)={\rm dim} (F)+{\rm dim} \left( t(C,F) \right)$ for any face $F$ of $C$.
    \mlabel{it:com2}
\end{enumerate}
To the first two properties correspond similar properties  for lattice cones.
\begin{enumerate}\setcounter{enumi}{3}
\item {\bf (Transitivity) } $t\left((C,\Lambda_C),(F,\Lambda_F)\right)=t\left( t\left( (C,\Lambda_C),(F^\prime,\Lambda_{F^\prime})\right), t\left( (F,\Lambda_F),(F^\prime,\Lambda_{F^\prime})\right)\right)$ if $(F^\prime,\Lambda_{F^\prime})$ is a face of $(F,\Lambda_F)$.
    \mlabel{it:trad}
\item {\bf (Compatibility with the partial order) } We have
$$\left\{(H,\Lambda_H)\preceq t\left((C,\Lambda_C),(F,\Lambda_F)\right)\right\} =\left\{(t((G,\Lambda_G),(F,\Lambda_F))\,|\, (F,\Lambda_F)\preceq (G,\Lambda_G)\preceq (C,\Lambda_C)\right\}.$$
 \mlabel{it:com1d}
 \end{enumerate}
\end{prop}

\begin{proof}
We first carry out the proof for ordinary cones.
\smallskip

\noindent
(\mref{it:tra}) For $F'\preceq F\preceq C$, we have ${\rm lin}(F')< {\rm lin}(F)< V_k$. Thus the inner product induces orthogonal decompositions
$$V_k={\rm lin}(F)\oplus {\rm lin}(F)^\perp,
\quad {\rm lin}(F)={\rm lin}(F')\oplus L.
$$
Therefore
$$V_k={\rm lin}(F')\oplus L\oplus {\rm lin}(F)^\perp,
\quad {\rm lin}(F')^\perp=L\oplus {\rm lin}(F)^\perp.
$$
By definition, we have
$$L={\rm lin}(t(F,F')),
\quad
L^\perp={\rm lin}(F')\oplus {\rm lin}(F)^\perp.
$$
This implies
$\pi_{F^\perp} = \pi_{t(F,F')^\perp}\, \pi_{{F'}^\perp}.$\footnote{The composition symbol $\circ$ will be suppressed throughout the paper.}  Thus we have
\begin{equation}
t(C,F)= \pi_{F^\perp}(C)
= (\pi_{t(F,F')^\perp} \,\pi_{{F'}^\perp})(C)
= \pi_{t(F,F')^\perp} (t(C,F'))
=t(t(C,F'),t(F,F')).
\mlabel{eq:tt}
\end{equation}
\smallskip

\noindent
(\mref{it:com1}) Assume that $F$ is defined by a linear form $u_F\in V^*$. Let $G$ be a face of $C$ containing $F$ that is defined by $u_G\in V^*$. Then $u_G|_F=0$.
But any element $u\in V^*$ with $u|_F=0$ induces an element $u\in (\linf)^*$. So we can view $u_G$ as an element in  $(\linf)^*$; it therefore defines a face $t(G,F)$ of $t(C,F)$. We can therefore define a map
\begin{eqnarray}
t(\bullet, F):\{\text{faces  } G \text{ of } C \text{ containing }  F\}& \to&  \{\text{faces } H \text{ of } t(C, F)\}, \mlabel{eq:face}\\
G & \mapsto & t(G,F)=t(C,F)\cap u_G^\perp. \notag
\end{eqnarray}

To check the bijectivity of $t(\bullet, F)$, first note that any face of $t(C,F)$ is defined by some element $u\in (\linf)^*$ which can be viewed as an element in $V^*$ that vanishes on $\lin(F)$. Hence $u$ defines a face $G$ of $C$ containing $F$. Thus $t(\bullet, F)$ is surjective.

Next for two different faces $G_1$, $G_2$ containing $F$  defined by $u_1$, $u_2\in V^*$, there are  vectors $v_1$ in $G_1$ and $v_2$ in $G_2$ such that $\langle u_1,v_2\rangle >0$ and $\langle u_2,v_1\rangle >0$. Thus $t(G_1,F)$ and $t(G_2,F)$ are different since the image of $v_1$ is not in $t(G_2, F)$ and the image of $v_2$ is not in $t(G_1,F)$. Hence the map $t(\bullet, F)$ is one-to-one. This gives the desired equation.
\smallskip

\noindent
(\mref{it:com2}) follows from the definition of $t(C,F)$ since $\lin (C)=\lin (F) \oplus \lin(t(C,F))$.
\smallskip

We next verify the properties for lattice cones. For Item~(\mref{it:trad}), by the definition of transverse lattice cones, the left hand side of the desired equation is
$$ t((C,\Lambda_C),(F,\Lambda_F))= (t(C,F), \Lambda_{t(C,F)}).$$
Similarly, the right hand side of the equation is
\begin{eqnarray*}
&&t(t((C,\Lambda_C),(F,\Lambda_F)),t((F,\Lambda_F),(F',\Lambda_{F'}))) \\
&=& t((t(C,F'),\Lambda_{t(C,F')}),(t(F,F'),\Lambda_{t(F,F')}))\\
&=& (t(t(C,F'),t(F,F')),\Lambda_{t(t(C,F'),t(F,F'))}).
\end{eqnarray*}
By Item~(\mref{it:tra}), the first components of the two sides agree. On the other hand,
$$\Lambda_{t(C,F)}= \pi_{F^\perp}(\Lambda_C) = (\pi_{t(F,F')^\perp}\,\pi_{{F'}^{\perp}})(\Lambda_C)
= \pi_{t(F,F')}(\Lambda_{t(C,F')}) = \Lambda_{t(t(C,F'),t(F,F'))}.
$$
Thus the second components of the two sides also agree. This proves Item~(\mref{it:trad}).

\smallskip

For Item~(\mref{it:com1d}), the bijection in Eq.~(\mref{eq:face}) can be enriched to the bijection
\begin{eqnarray*}
t(\bullet, F):\{\text{faces  } (G,\Lambda_G) \text{ of } (C,\Lambda_C) \text{ containing }  (F,\Lambda_F)\}& \to&  \{\text{faces } (H,\Lambda_H) \text{ of } (t(C, F),\Lambda_{t(C,F)} \}, \mlabel{eq:faced}\\
(G,\Lambda_G) & \mapsto & (t(G,F),\Lambda_{t(G,F)}) \notag
\end{eqnarray*}
since we know from Proposition \ref{prop:latticetrans} that we can make sense of the lattice $\Lambda_{t(G,F)}$ of $t(G,F)$. \end{proof}

\subsection {The coalgebra of {\ltcone}s}
Let us now introduce the concept of a connected coalgebra similar to that of a connected bialgebra~\cite{Ma}.
See also~\cite[\S~2.3]{Gub}.

\begin{defn}\mlabel{def:diffconilpotent}
{\rm Let $(\coalg,\Delta)$ be a coalgebra over a field $\bfk $ with counit $\e:\coalg\to \bfk$. It is called
\begin{enumerate}
\item {\bf cograded} if there is a grading  $\coalg = \bigoplus_{n\geq 0}\coalg^{(n)}$ such that
$$ \Delta(\coalg^{(n)})\subseteq \bigoplus_{p+q=n} \coalg^{(p)}\ot \coalg^{(q)}, \quad n\geq 0.$$
Elements in $\coalg^{(n)}$ are said to have {\bf degree $n$}.
\item
{\bf coaugmented} if there is a linear map $u:\bfk \to \coalg$, called the {\bf coaugmentation}, such that $\e\, u=\mathrm{id}_{\bfk}$.
\item
{\bf connected} if $ \coalg^{(0)}=\bfk\,u(1).$
\end{enumerate}
}
\end{defn}

With the coaugmentation $u$, $\coalg$ is canonically isomorphic to $\ker \e \oplus \bfk u(1)$.
The proof of the following lemma is similar to the one for the case of connected bialgebras.

\begin{lemma}
Let $(\coalg,\Delta)$ be a cograded, coaugmented, connnected coalgebra. Then
$$ \ker \e = \bigoplus_{n\geq 1} \coalg^{(n)}.$$
Further the
{\bf reduced coproduct}
$$\bar{\Delta}: \ker \e\to \ker \e\ot \ker \e \quad x\mapsto \Delta(x)-x\ot u(1) -u(1) \ot x \quad \text{for all }  x\in \coalg,$$
is well defined and $\bar{\Delta}^{m}(\coalg^{(n)})=0$
for $m\geq n\geq 1$, where
$\bar{\Delta}^{m}, m\geq 2,$ is defined by the recursion $\bar{\Delta}^{m} = (\id \ot \bar{\Delta}^{m-1})\,\bar{\Delta}.$
\end{lemma}
The last condition is called the {\bf conilpotency} of $\Delta$~\mcite{LV}.

We now equip the linear space $\Q \cc$   of {\ltcone}s with the linear maps
\begin{equation}
\Delta: \Q\cc \longrightarrow \Q\cc \otimes \Q\cc, \quad
(C,\Lambda _C)\longmapsto \sum_{F\preceq C} (t(C,F), \Lambda _{t(C,F)}) \ot (F, \Lambda _F),
 \mlabel{eq:coproduct}
\end{equation}
\begin{equation}
\e: \Q\cc \longrightarrow \Q, \quad
(C,\Lambda _C)\longmapsto \left \{\begin{array}{ll} 1, & C=\{0\}, \\
0, & C\neq \{0\}, \end{array} \right .
\mlabel{eq:counit}
\end{equation}
and
 \begin{equation}
 u:\Q\longrightarrow \Q\cc, \quad 1 \longmapsto (\{0\},\{0\} ).
  \mlabel{eq:unit}
 \end{equation}

\begin{thm} The quadruple $(\Q\cc ,\Delta,\e,u)$ with $\Delta$, $\e$ and  $u$   as in Eqs.~$($\mref{eq:coproduct}$)$, $($\mref{eq:counit}$)$ and $($\mref{eq:unit}$)$,
defines a connected cograded coaugmented coalgebra with the grading
\begin{equation}
\QQ\cc =\bigoplus_{n\geq 0} \QQ \cc ^{(n)},
\mlabel{eq:grading}
\end{equation}
where
$$\cc ^{(n)}:= \left \{ (C,\Lambda_C)\in \cc \,\big|\, \dim\,C=n\right\},\quad n\geq 0.$$
\mlabel{thm:Hopfoncones}
\end{thm}

\begin{proof}
Let $(C, \Lambda_C)$ be a \ltcone \ in $\cc$, where $C \subset V_k$. On the one hand, we have
\begin{eqnarray*}(\id\otimes \Delta)\Delta(C,\Lambda_C) &=&  \sum_{F\preceq C} (\id\otimes \Delta) ( ( t(C,F), \Lambda _{t(C,F)})\otimes (F, \Lambda _F))\notag\\
&=& \sum_{F' \preceq F\preceq C} (t(C,F), \Lambda _{t(C,F)})\otimes (t(F,F'),\Lambda_{t(F,F')}) \otimes (F',\Lambda _{F'}).
\mlabel{eq:coar}
\end{eqnarray*}

On the other hand,
\begin{eqnarray*}
(\Delta\otimes \id)\Delta(C, \Lambda _C)&=&\sum_{F' \preceq C} (\Delta \otimes \id) ( (t(C,F'), \Lambda _{t(C,F' )})\otimes
(F', \Lambda_{F' }))\\&=& \sum_{F'\preceq C}\sum_{H\preceq t(C,F')} \left(  (t(t(C,F'),H), \Lambda _{t(t(C,F'),H)})\otimes (H,\Lambda _H)  \right)\otimes (F', \Lambda _{F'}).
\end{eqnarray*}
For $H\preceq t(C,F')$, by Proposition~\mref{pp:transversecone}.(\mref{it:com1}) and (\mref{it:com1d}), there is $F\preceq F'\preceq C$ such that
$t(F,F')=H$ and $\Lambda _{t(F,F')}=\Lambda _H$.
Further,  since $\pi_{H^\perp}\,\pi_{F'^\perp}=\pi _{t(F,F')^\perp}\pi _{F'^\perp}=\pi _{F^\perp}$, we have
$$t(t(C,F'),H)=\pi_{H^\perp}\left(t(C,F')\right) =\pi_{H^\perp}\left(\pi_{F'^\perp}\left(C\right)\right)
= \pi_{F'^\perp}\left(C\right)
=t(C,F).$$
Similarly,
$\Lambda _{t(t(C,F'),H)}=\Lambda_{t(C,F)}.$
This proves the coassociativity.
It then follows from the definitions, that $\e$ is a counit for $\Delta$ and that $u$ yields a coaugmentation. Furthermore by Proposition~\mref{pp:transversecone}.(\mref{it:com2}), the grading in Eq.~(\mref{eq:grading}) turns $(\QQ\cc,\Delta,\e)$ into a cograded coalgebra. Since $\cc^{(0)}=(\{0\}, \{0\})$, $\QQ\cc$ is connected.
\end{proof}

\section{Subdivision properties} In this section we  study the behavior of    linear maps on the space of lattice cones  with respect to subdivisions.

\subsection{Subdivisions of lattice cones}
\begin{defn}
\begin{enumerate}
\item  A {\bf subdivision} of a cone $C$ is a set $\{C_1,\cdots,C_r\}$ of cones such that
\begin{enumerate}
\item[(i)] $C=\cup_{i=1}^r C_i$,
\item[(ii)] $C_1,\cdots,C_r$ have the same dimension as $C$, and
\item[(iii)] $C_1,\cdots,C_r$ intersect along their faces, i.e.,  $C_i\cap C_j$ is a face of both $C_i$ and $C_j$.
\end{enumerate}
\item
A {\bf subdivision} of a \ltcone  \ $(C, \Lambda _C)$  is a set of {\ltcone}s $\{(C_i, \Lambda _{C_i})\,|\, 1\leq i\leq r\}$ such that $\{C_i\,|\,1\leq i\leq r\}$ is a subdivision of $C$ and $\Lambda_{C_i}=\Lambda _C$ for all $1\leq i\leq r$.
\item
A cone or lattice cone is called its own {\bf trivial subdivision.}
\end{enumerate}
\end{defn}

\begin{defn}
Let $\lC$ be a simplicial \ltcone \ in $\cc_k$ and let $\conefamilyc =\{ (C_1, \Lambda_C),\cdots, (C_r,\Lambda_C) \}$ be a  subdivision of $\lC$ into simplicial cones.
Let $\mathcal{F}^o(\conefamilyc )$ denote the set of  faces of $C_1,\cdots,C_r$ that are not contained in any proper face of $C$, that is, those faces of $C_1,\cdots,C_r$ that intersect with the interior of $C$.
\mlabel{de:subd}
\end{defn}

Just as for ordinary cones, we have the following property.
\begin{prop}
Any \ltcone \ can be subdivided into smooth {\ltcone}s.
\end{prop}

\begin{proof} For a given \ltcone $(D, \Lambda _C)$ in  a simplicial subdivision of a \ltcone $\lC$  with its primary generating set  $\{v_1, \cdots, v_{n}\}$,  we write $v_i=\sum\limits_{j=1}^{n} a_{ij}u_j$, $a_{ij}\in \ZZ$, $i=1, \cdots , n$, where $\{u_1, \cdots , u_{n}\}$  is  a basis of $\Lambda_C$. The absolute value of the determinant $w_D=|v_1, \cdots, v_{n}|:=|\det(a_{ij})|$ lies  in $\Z _{\ge 1}$ and is independent of the choice of a basis $\{u_1,\cdots,u_n\}$ of $\Lambda_C$. Further $w_D$ is equal to one  if and only $(D, \Lambda _C)$ is smooth.

We now prove the proposition by contradiction. Suppose $(C,\Lambda_C)$ is a  lattice cones that cannot be subdivided into smooth lattice cones.  Then for any simplicial subdivision $\conefamilyc:=\{(C_i,\Lambda_C)\}$ of $(C,\Lambda_C)$, we have
$$ w_{{\conefamilyc}}:=\max \{w_{C_i}\}>1 \quad \text{ and } \quad n_{\conefamilyc}:=\max |\{i\,|, w_{C_i}=w_{\conefamilyc}\}|\geq 1.$$

Choose a simplicial subdivision ${\conefamilyc}$ of $(C,\Lambda_C)$ with $w_{\conefamilyc}$ minimal and then among those, one with $n_{\conefamilyc}$ minimal.
We will construct a subdivision of $\lC$ that refines $\conefamilyc$. Let $D=\cone{v_1, \cdots, v_{n}}$ be a cone in ${\conefamilyc}$ with $w_D=w_{\conefamilyc}$. Since $w_D>1$, the lattice cone $(D, \Lambda_C)$ is not smooth. So $\{v_1, \cdots v_n\}$ is not a lattice basis of $\Lambda_C \cap D$. Note that the set $\{v_1,\cdots,v_{n}\}\cup \left(\left (\sum\limits_{i=1}^{n} [0,1)v_i\right) \cap \Lambda_C\right)$ spans $\Lambda_C\cap D$ as a monoid.  So there is a vector $0\neq v_D=\sum\limits_{i=1}^n c_i v_i\in \Lambda_C$ with $c_i\in [0,1)$ rational.

Reordering $v_i$, we can assume that $c_i\not=0$ for $i=1, \cdots, k,$ and $c_i=0$ for $i=k+1, \cdots , n$. We now use the vector $v_D=\sum\limits_{i=1}^{k} c_i v_i$ to subdivide the cones. Let $C_i=\cone{v_1, \cdots, v_{k}, v^{i}_{k+1}, \cdots, v^{i}_{n}}$, $i=1, \cdots , s$, be all the cones arising in the subdivision $\conefamilyc$ that contain $\cone{v_1, \cdots, v_k}$ as a face, with $C_1=D$. Then the set of cones
$$\{C_i, i>s\}\cup \{C_{ij}:=\cone{v_1,\cdots, \check{v_j}^D,\cdots,v_{k},v^{i}_{k+1}, \cdots, v^{i}_{n}}\,|\, j=1,\dots, k, \ i=1, \cdots, s\},$$
where $\check{v_j}^D$ means  $v_j$ has been replaced by $v_D$, yields a new subdivision ${\conefamilyc}'$ of $C$.

For elements in ${\conefamilyc}'$, the numbers $w_{C_i},i>s$  coincide. For $i=1,\cdots,s$ and $j=1, \cdots, k$,
$$|v_1,\cdots,\check{v_j}^D,\cdots,v_k,v^{i}_{k+1}, \cdots, v^{i}_{n} | =c_j|v_1,\cdots,v_k, v^{i}_{k+1}, \cdots, v^{i}_{n}|<|v_1,\cdots,v_k, v^{i}_{k+1}, \cdots, v^{i}_{n}|=w_{C_i}.$$
So $w_{C_{ij}}<w_{{\conefamilyc}}$. Therefore either $w_{{\conefamilyc}'}<w_{\conefamilyc}$, or $w_{{\conefamilyc}'}=w_{\conefamilyc}$ and $n_{{\conefamilyc}'}<n_{{\conefamilyc}}$. This gives the desired contradiction.
\end{proof}

We collect the following facts before introducing more concepts on subdivisions of cones.
\begin{lemma} Let $C$ be a cone and let $\conefamilyc=\{C_i\}$ be a subdivision of $\ C$.
\begin{enumerate}
\item
If $F\preceq C$ and $\rdim (F)=\rdim (C)$, then $F=C$.
\mlabel{it:bas1}
\item
If $F \preceq C$, then $F\cap C_i \preceq C_i$.
\mlabel{it:bas2}
\item
For any cone $G$ inside $C$, the set $\conefamilyc (G)=\{C_i\cap G\,|\, \dim C_i\cap G=\dim G \}$ is a subdivision of $G$.
\mlabel{it:bas3}
\mlabel {lem:IndSub}
\end{enumerate}
\mlabel{lem:bas}
\end{lemma}

\begin {proof}
(\mref{it:bas1}) Let $F$ be defined by a linear functional $u$. Thus $u|_C\ge 0$ and $F=C\cap u^\perp$. Since $u(F)=0$ we have  $u({\rm lin}(F))=0$. But $\lin(F)=\lin(C)$ since  $\dim (F)=\dim (C)$. Thus $u(C)=0$, forcing $F=C$.
\smallskip

\noindent
(\mref{it:bas2})
If $F$ is defined by $u$, then $F\cap C_i=C\cap u^\perp\cap C_i=C_i\cap u^\perp$. So it is a face of $C_i$.
\smallskip

\noindent
(\mref{it:bas3}) Let $D$ be the union of $C_i\cap G$ with $\dim C_i\cap G<\dim G$. Then $G\backslash D$ is dense in $G$. Thus as its superset, the union of the cones $C_i\cap G$ with $\dim C_i\cap G=\dim G$ is dense in $G$ and hence is $G$.  These cones intersect along their faces and hence provide a subdivision of $G$.
\end{proof}

Given a subdivision  $\conefamilyc:=\{C_1, \cdots, C_n\}$ of $C$, set
\begin{equation}\label{eq:P}
\calp:=\calp_{\conefamilyc}:=\{\text{non-zero, proper face of some } C, C_1,\cdots,C_n\}\end{equation}
and
$$\calp_C:=\{F\in \calp\,|\, F\preceq C\}.$$

Denote $[n]:=\{1,\cdots,n\}$. For $I\subset [n]$, let
$$C_I:=\cap_{i\in I} C_i\quad \text{and} \quad \calt: =\{C_I\ |\ \emptyset \neq I\subset [n]\}.
$$
For a face $F\in \calp$ , set
\begin{equation}\label{eq:JF}J(F):=\{i\in [n] \ | \ F\preceq C_i \} \quad  \text{and}\quad j(F):=|J(F)|.
\end{equation}
Note that $H\preceq F$ implies $J(H)\supset J(F)$.

For any subset $\calq$ of $\calp$ and $i\geq 0$, we further set
\begin{equation}
 \calq_i:=\{F\in \calq\,|\, j(F)=i\}, \quad \calq_{\geq i}:=\{F\in \calq\,|\,j(F)\geq i\}.
\label{eq:degi}
\end{equation}
In particular this notation applies to $\calp_C$.

\begin {defn} Let $\{C_i\}$ be a subdivision of $C$. A proper face of a $C_i$ is called a {\bf subdivision induced face (SIF)} if it arises as a cone in a nontrivial subdivision of some face of $C$.
\end {defn}

Distinguishing between  faces  induced and not induced by a subdivision,
\begin{equation}
\calp_\civ:=\{F \in \calp \ |\ F \ {\rm is \ an \ SIF} \} \quad \text{and} \quad \calp_\cv:=\{F \in \calp \ |\ F\npreceq C,\, \, F \ {\rm is \ not \ an \ SIF} \}.
\mlabel{eq:isf}
\end{equation}
yield a partition
\begin{equation} \calp=\calp_\ci\coprod \calp_\cii\coprod \calp_\ciii \coprod \calp_\civ \coprod \calp_\cv
\mlabel{eq:disj}
\end{equation}
of $\calp$ into the five subsets of cones arising respectively as proper nonzero faces
\begin{itemize}
\item of $C$ that are not faces of any   cone in the subdivision,
\item of $C$ that are  faces of exactly one  cone in the subdivision,
\item of $C$ that are  faces of  at least two  cones in the subdivision,
\item of some $C_i$ and arising from a nontrivial subdivision of some face of $C$,
\item of some $C_i$ but not of $C$ and not arising from a nontrivial subdivision of any face of $C$.
\end{itemize}

\begin{ex} For the subdivision $\{\cone{e_1,e_1+e_2,e_3},\cone{e_2,e_1+e_2,e_3}\}$ of the cone $\cone{e_1,e_2,e_3}$, we have $$\calp_\ci=\{\langle e_1,e_2\rangle\}; \calp_\cii=\{\langle e_1\rangle,\langle e_1,e_3\rangle, \langle e_2\rangle, \langle e_2, e_3\rangle\}; \calp_\ciii =\{\langle e_3\rangle\};$$
$$ \calp_\civ =\{\langle e_1,e_1+e_2\rangle, \langle e_2,e_1+e_2\rangle \}; \calp_\cv=\{\langle e_1+e_2\rangle, \langle e_1+e_2, e_3\rangle\}.$$
\end{ex}

\begin{lemma}
\begin{enumerate}
\item
The relation
$$R:=\{(F,G)\in \calp_\civ\times \calp_\ci\,|\, F\subset G, \ {\rm dim }F={\rm dim }G \}$$
defines a surjective map
\begin{eqnarray}
\alpha: \calp_\civ &\to &\calp_\ci\\
F &\longmapsto& G,\,\ (F,G)\in R.\nonumber
\mlabel{eq:alpha}
\end{eqnarray}
\mlabel{it:alpha1}
\item
For each $G\in \calp_\ci$, the set $\alpha^{-1}(G)$ is a subdivision of $G$.
\mlabel{it:alpha2}
\end{enumerate}
\mlabel{lem:alpha}
\end{lemma}

\begin{proof}
(\mref{it:alpha1}) Let $F\in\calp_\civ$. Then $F$ arises in a subdivision of a face $G$ of $C$, but is not equal to $G$. Such a face $G$ of $C$ is unique: if $F $ is contained in $G_1$ and $G_2$,
then
$$\dim\,G_i\geq \dim\,(G_1\cap G_2)\geq \dim\,F\geq \dim\,G_i, i=1,2.$$
Thus $\dim(G_1\cap G_2)=\dim G_1=\dim G_2$. Also $G_1\cap G_2 \preceq G_1$, $G_1\cap G_2 \preceq G_2$. So $G_1=G_1\cap G_2=G_2$ by Lemma \mref {lem:bas} (\mref {it:bas2}). Further $G$ lies  in $\calp_\ci$ for, if $G$ were contained in some $C_i$, then $G=G\cap C_i$ would be a face of $C_i$ by Lemma~\mref{lem:bas}.(\mref{it:bas2}), leading to a contradiction.
Thus we obtain a map
$$\alpha : \calp_\civ \to \calp_\ci
$$
sending $F\in \calp_\civ$ to the unique face $G$ of $C$ above. The map is surjective in view of Lemma~\mref{lem:bas}.(\mref{it:bas3}).
\smallskip

\noindent
(\mref{it:alpha2}) For $G\in \calp_\ci$, $\alpha ^{-1}(G)$ gives the subdivision of $F$ induced by $\{C_i\}$ as explicited in Lemma~\mref{lem:bas}.(\mref{it:bas3}).
\end{proof}

On the grounds of this lemma we  introduce further useful notations.
For $G\in \calp_\ci$, let
$$\alpha ^{-1}(G)=\{F_1^G, \cdots, F^G_{\ell(G)}\}, \quad \ell(G)=|\alpha ^{-1}(G)|.$$
For $k\ge 1$, let
$\alpha ^{-1}(G)_k=\{F^G\in \alpha ^{-1}(G)\,|\,j(F^G)=k\}$.

\subsection{Induced subdivisions on transverse cones}
\mlabel{sss:transubd}
We now study how a subdivision of a cone induces a subdivision on a transverse cone. We first recall the following fact.

\begin {lemma}
{\bf (Separation Lemma~\cite {Fu})}
For cones $C_1$ and $C_2$ with $C_1\cap C_2\preceq C_1$ and $C_1\cap C_2\preceq C_2$, there exists a linear function $u$ such that $u|_{C_1}\ge 0$, $u|_{C_2}\le 0$ and $C_1\cap C_2=C_1\cap u^\perp=C_2\cap u^\perp$.
\mlabel {lem:Sep}
\end{lemma}
Applying the separation lemma to transverse cones yields
\begin {lemma}\begin{enumerate}
 \item
Let $C_1$ and $C_2$ be  cones and let $F:=C_1\cap C_2$. If $F\precneqq C_1$ and $F\precneqq C_2$, then $t(C_1,F)\not =t(C_2, F)$.
\mlabel{it:two}
\item
Let $\{C_i\}$ be a subdivision of $C$ and let $H\in \calp$. Then the cones $\{t(C_i, H)\,|\, i\in J(H)\}$ are distinct.
\mlabel{it:dist}
\end{enumerate}
\mlabel{lem:dist}
\end{lemma}

\begin {proof}
(\mref{it:two}) Take the linear function $u$ in Lemma \mref {lem:Sep}. By  assumption, there are $c_1\in C_1$ and $c_2\in C_2$ such that $u(c_1)>0$ and $u(c_2)<0$. Thus $C_1$ and $C_2$ are distinct. Since $u$ vanishes on $F$, $u$ descends to a linear function $\bar{u}$ on the  space ${\rm lin} (F)^\perp$. Further, we have $\bar{u}(c_1+\lin F)=u(c_1)$ and $\bar{u}(c_2+\lin F)=u(c_2)$. Since $c_1+\lin F\in t(C_1,F)$ and $c_2+\lin F\in t(C_2, F)$, the conclusion follows.
\smallskip

\noindent
(\mref{it:dist}) Given $H\in \calp$, let $m\neq n$ be in $J(H)$. Since $\{C_i\}$ is a subdivision of $C$, the condition in Item~(\mref{it:two}) is satisfied. Thus $t(C_m,C_m\cap C_n)$ and $t(C_n,C_m\cap C_n)$ are distinct. But $t(C_m,C_m\cap C_n)$ and $t(C_n,C_m\cap C_n)$ are quotients of $t(C_m,H)$ and $t(C_n,H)$ respectively modulo $\lin(C_m\cap C_n)$. Hence $t(C_m,C_m\cap C_n)$ and $t(C_n,C_m\cap C_n)$ are also distinct.
\end{proof}

\begin{lemma}
Let $\{C_1, \cdots, C_n\}$ be a subdivision of $C$ and let $F$ be a face of some $C_i$.
\begin{enumerate}
\item The cones
$\{t(C_i,F)\,|\,i\in J(F)\}$ are distinct and form a subdivision of $t(C,F)$. Here by $t(C,F)$ we mean the projection of $C$ in ${\rm lin}(F)^\perp$ even if $F$ is not a face of $C$. In particular, if $F$ is in $\calp_\cii$, so that $F\preceq C$ and $J(F)=\{C_{i_0}\}$, then $t(C,F)=t(C_{i_0},F)$.
\mlabel{it:trpa}
\item For $I\subset J(F)$ we have
$$\bigcap _{i\in I}t(C_i, F)=t(C_I,F).
$$
\mlabel{it:trpb}
\item If $F\in \calp_\cv$, that is, if $F$ is a face of a $C_i$ but neither a face of $C$ nor an SIF, then $t(C,F)$ contains a line.
\mlabel{it:trpc}
\end{enumerate}
\label{lem:transversepartition}
\end{lemma}
\begin{proof}
(\mref{it:trpa}) Clearly,  $t(C,F)=\bigcup_{i=1}^{n}t(C_{i},F)$.  We first need to  prove that $t(C,F)=\bigcup_{i\in J(F)}t(C_i,F).$

For any $x$ in  $C_{i}$ such that $i \notin J(F)$, let $x_{0}\neq x$ be any point in the relative interior of $F\cap C_{i}$.
The line segment $[x_{0},x]$ lies in $C_i$ and hence in $C$. It  intersects $C_j$ for some $j\in J(F)$ at a point $y\not = x_0$ for otherwise, $x_0=y\in C_j$  would lie on a face of $C_j$, $j\not\in J(F)$, so the relative interior of $F\cap C_{i}$ would lie on a face of $C_j$ contradicting the assumption on $J(F)$.
As an element of $V_k$, we have
$$x=\frac{||x-x_{0}||}{||y-x_{0}||}(y-x_{0})+x_{0}.$$
Therefore
$$\pi_{F^{\perp}}\,(x) =\pi_{F^{\perp}}\,\left(\frac{||x-x_{0}||}{||y-x_{0}||}y\right),$$
which is an element of $\bigcup_{i\in J(F)}t(C_i,F)$, as required.

We next prove that the cones $t(C_{i},F)$, $i\in J(F)$, only intersect along their faces.
If distinct cones $C_i$ and $C_j, i, j\in [n]$ have a common face $F$, then $F\subset C_i\cap C_j$. By Lemma~\mref {lem:Sep}, there exists a linear function $u$, such that $u|_{C_i}\ge 0$, $u|_{C_j}\le 0$ and $C_i\cap C_j=C_i\cap u^\perp=C_j\cap u^\perp$. Then for $x_i\in C_i$, $x_j\in C_j$, if $\pi_{k,F^\perp}(x_i)=\pi_{k,F^\perp}(x_j)$, then $u(x_i)=u(x_j)$, so $u(x_i)=u(x_j)=0$. Therefore $x_i\in C_i\cap C_j$ and $x_j\in C_i\cap C_j$. So $t(C_i,F)\cap t(C_j,F)=t(C_i\cap C_j,F)$. This gives what we need since by Proposition~\mref{pp:transversecone}.(\mref{it:com1}), the right hand side is a face of the two cones on the left hand side. Now assertion (\mref{it:trpa}) follows from  Lemma~\mref{lem:dist} (\mref{it:dist}).

\smallskip

\noindent
(\mref{it:trpb}) We proceed by induction on $|I|$. The case  $|I|=1$ is trivial. Reordering the cones if necessary, we  assume that the desired equation holds for $I=[k]$ with $k\geq 1$, and aim to prove it  when $I=[k+1]$. If $C_{[k]}\subset C_{k+1}$, then $C_{[k]}=C_{[k+1]}$ and  $t(C_{[k]},F)\subseteq t(C_{[k+1]},F)$. Thus
$$ t(C_{[k+1]},F)=t(C_{[k]},F)=t(C_{[k]},F)\cap t(C_{k+1},F)=\cap_{i\in [k+1]} t(C_i,F).$$
If $C_{[k]}\not \subset C_{k+1}$, then we can apply the same argument as in the previous item with $C_m, C_n$ replaced by $C_{[k]}, C_{k+1}$ since the argument only requires the two cones to be  different and to intersect along their faces. It follows that $t(C_{[k]},F)\cap t(C_{k+1},F)=t(C_{[k]}\cap C_{k+1}, F)$, as needed to complete the induction.

\smallskip
\noindent
(\mref{it:trpc})
We prove the property by induction on $k={\rm dim}(C)$.

For $k=1$, there is nothing to prove. Let us assume that  the claim  holds for $k=n$ and let us prove it for $k=n+1$.

Let $F$ be a proper face of a $C_i$ but is neither a face of $C$ nor an SIF. Since $F$ is not a face of $C$, it is either not contained in any proper face of $C$, or it is properly contained in a proper face of $C$.

Assume $F$ is not contained in any proper face of $C$. Then there exists a point $x_{0}$ of $F$ that is in the relative interior of $C$.
Since ${\rm dim}(F)<{\rm dim }(C)$, there is a point $x_1\in C$, $x_1\not =x_0$,  such that $\pm x_1+x_0\in  {\rm lin}^{\perp}(F)$. Therefore $t(C,F)$ contains the line $\RR x_1$.

Now assume that $F$ is properly contained in a proper face of $C$. Let $G$ be a face of $C$ that contains $F$ and has minimal dimension with this property. This $G$ is unique; indeed  if both $G_1$ and $G_2$ are faces of $C$ containing $F$ and having minimal dimension, then so is $G_1\cap G_2$, which  by Lemma~\mref{lem:bas}.(\mref{it:bas2}) implies that $G_1=G_2=G_1\cap G_2$, leading to a contradiction. Now $F$ is neither a face of $G$ nor an SIF, so by the induction hypothesis, $t(G,F)$ contains a line. Then $t(C,F)\supseteq t(G,F)$ contains a line.
\end{proof}

\subsection{Compatibility of the convolution product with subdivisions}
Let ${\mathcal C}$ a class of sets stable under finite intersections and finite unions. A map $\varphi$ on ${\mathcal C}$   with values in a commutative algebra $A$ is said to satisfy the valuation property if
$$ \varphi(A\cup B)+ \varphi(A\cap B) = \varphi(A)+\varphi(B) \quad \text{for all } \, A,B\in\mathcal{C}.$$  A straightforward induction shows that a map obeys the valuation property if and only if it satisfies the following compatibility with unions:
\begin{equation}\label{eq:compatibilityunions}\phi(\cup_{i=1}^n A_i)=\sum_{\emptyset\neq I\subset [n]}(-1)^{\vert I\vert-1}  \phi(A_I)\quad \text{for all } A_1,\cdots, A_n\in {\mathcal C},\end{equation}
where we have set  $A_I:= \cap _{i\in I} A_i$. For the cardinal on finite sets, Eq.~(\ref{eq:compatibilityunions}) amounts to the inclusion-exclusion principle.

We extend the valuation property of the form in Eq.~(\ref{eq:compatibilityunions}) to subdivision properties for maps on lattice cones. Notice that that the set of lattice cones is only  equipped with a partial intersection and a  partial union.

\begin {defn} \label{defn:ValP} A linear map $\phi$  on $ \QQ \frakC $  with values in a commutative algebra has
\begin {itemize}
\item the {\bf \ValP} if for a \ltcone $\lC$ and its subdivision $\conefamilyc=\{(C_i, \Lambda_{C_i} )\}_{i=1,\cdots ,n}$,
\begin{equation}
\phi \lC=\sum _{\emptyset \neq I\subseteq [n]}(-1) ^{|I|-1}\phi (C_I, \Lambda_{C_I}).
\mlabel{eq:VaSubP}
\end{equation}
\item
the {\bf \SSubP} if for a \ltcone $\lC$ and its subdivision $\conefamilyc=\{(C_i, \Lambda_{C_i} )\}_{i=1,\cdots ,n}$,
\begin{equation}
\phi \lC=\sum _{F\in \mathcal {F}^o (\conefamilyc) }\phi (F, \Lambda_F).
\label{eq:SSubP}
\end{equation}
\item
the {\bf \ISubP} if for a \ltcone $\lC$ and its subdivision $\conefamilyc=\{(C_i, \Lambda_{C_i} )\}_{i=1,\cdots ,n}$,
\begin{equation}
\phi \lC=\sum _{i=1 }^n\phi (C_i, \Lambda_{C_i}).
\label{eq:ISubP}
\end{equation}
\end{itemize}
\end{defn}

The \ValP is closely related to \SSubP. For a linear map $\phi: \QQ \frakC \to {\mathcal A}$, we define
the map $\phi^c: \QQ \frakC \to {\mathcal A}$
by
$$\phi ^c \lC :=\sum _{F\preccurlyeq C}\phi (F,\Lambda_{F}).
$$

Then we have

\begin {prop}
\mlabel {prop:ValPSSubP}
A linear map $\phi $ has the \SSubP if and only if $\phi ^c$ has the \ValP.
\end{prop}

\begin{proof} For a \ltcone $\lC$ and its subdivision $\conefamilyc=\{(C_i, \Lambda_{C_i} )\}_{i=1,\cdots ,n}$,
\begin{eqnarray*}
\phi^c(C)-\sum_{\emptyset\neq I\subseteq [n]} (-1)^{|I|-1}\phi^c(C_I)
&=& \sum_{F\preccurlyeq C}\phi(C) - \sum_{\emptyset\neq I\subseteq [n]}(-1)^{|I|-1}\sum_{F\preccurlyeq C_I}\phi(F) \\
&=& \sum_{F\preccurlyeq C} \phi(F) - \sum_{F\in \mathcal{F}^c(\conefamilyc)} \left(\sum_{\emptyset\neq I\subseteq J(F)}(-1)^{|I|-1}\right) \phi(F) \\
&=& \sum_{F\preccurlyeq C}\phi(F) - \sum_{F\in \mathcal{F}^c(\conefamilyc)} \phi(F) \\
&=& \phi(C)-\sum_{F\in \mathcal{F}^o(\conefamilyc)}\phi(F) +\sum_{G\precnsim C}\left(\phi(G)-\sum_{F\in\conefamilyc _G }\phi(F)\right) \\
&=&\phi(C)-\sum_{F\in \mathcal{F}^o(\conefamilyc)}\phi(F) +\sum_{G\precnsim C}\left(\phi(G)-\sum_{F\in\mathcal{F}^o(\conefamilyc(G))}\phi(F)\right).
\end{eqnarray*}
Here the third equation follows from $\sum\limits_{Y\subseteq X}(-1)^{|Y|}=0$ for a finite set $X$; the fourth equation follows from
$$\conefamilyc _G:=\{F \in \mathcal {F} ^c(\conefamilyc) \ | \ F \preccurlyeq G, \ F\ {\rm is\ not\ contained\ in\ any\ proper\ face\ of\ }G\},
$$
and the fact that $\mathcal {F}^c(\conefamilyc)$ is a disjoint union of $\conefamilyc _G, G\preccurlyeq C$; the fifth equation is a consequence of Lemma~\ref{lem:bas} (\ref{it:bas3}).

Now if $\phi$ has the \SSubP, then the right hand side is zero, so  the left hand is zero and $\phi^c$ has the \ValP.

Conversely, if $\phi^c$ has the \ValP. Then the left hand side, and hence the right hand side, is zero for all $C$. Note that for a one dimensional cone, the second sum on the right hand side is zero, showing that $\phi$ has the \SSubP for one dimensional cone. Then by an induction on the  dimension, $\phi$ is \SSubP for all cones.
\end{proof}

We now state our main theorem on \ValP of convolution quotient of linear maps on lattice cones.

\begin {theorem}
\mlabel {thm:IsubP}\mlabel{thm:0-sub}
Let $\phi$ and $\psi$ be linear maps on $\QQ\cc$ with values in a commutative algebra A that satisfy the following properties:
\begin{enumerate}
\item $\phi$ and $\psi$ satisfy the \ValP  property;
\item $\phi(\{0\},\{0\})=\psi (\{0\},\{0\})=1$;
\item for a \ltcone $\lC$ that is not strongly convex, $\phi (\lC)=\psi (\lC)=0$.
\end{enumerate}
Then the convolution quotient $\chi:=\phi ^{*(-1)}*\psi$ has the \ISubP.
\end{theorem}

We introduce more notations and preliminary results before actually proving the theorem. In the coalgebra $\QQ \frakC$, we set
\begin{equation}
\Delta' \lC:=\Delta \lC-(\{0\},\{0\})\otimes \lC-\lC \otimes (\{0\},\{0\}).
\label{eq:scoprod}
\end{equation}
We also use $*'$ to denote the restricted product of the convolution product in the space ${\mathcal L}(\QQ \frakC,A)$ of linear maps built from $\Delta '$, that is
$$\phi _1*'\phi _2=m_A\,(\phi _1\otimes \phi _2)\,\Delta ',$$
where $m_A$ is the multiplication of $A$. Then  $\phi_1*\phi_2=\phi_1*'\phi_2+\phi_1+\phi_2$.

\begin {lemma}
The map $\chi$ satisfies the recursive formula
\begin{equation}
\chi =\psi -\phi -\phi *'\chi.
\label{eq:indchi}
\end{equation}
\label {lem:IndChi}
\end{lemma}
\begin{proof} The right hand side of the equation gives
$$ \psi-\phi - \phi*'(\phi^{\ast (-1)}\ast \psi)
=\psi -\phi -\phi\ast (\phi^{\ast (-1)}\ast\psi) +\phi+\phi^{\ast(-1)}\ast \psi
=\chi,$$
as needed.
\end{proof}

Now, with $\calp$ as defined in Eq.~(\ref{eq:P}), we have
$$\Delta '((C, \Lambda_C)-\sum_{i=0}^n(C_i, \Lambda_{C_i}))=\sum _{F\in \calp}c(F)\otimes (F, \Lambda_F),
$$
where $c(F):=\sum\limits _{i=0}^nc_i(F)$
while, with the convention that $C_0=C$,
$$
c_i(F):=\left\{ \begin {array}{ll}
(t(C_i,F), \Lambda_{t(C_i,F)}), &i=0, F\preceq C_0,\\
-(t(C_i,F), \Lambda_{t(C_i,F)}), &i=1,\cdots,n, F\preceq C_i, \\
0,& F\npreceq C_i.
\end{array}\right.
$$

Then by Eq.~(\mref{eq:indchi}), we have
\begin{equation}
\chi \left((C, \Lambda_C)-\sum_{i=1}^n(C_i, \Lambda_{C_i})\right)= (\psi-\phi)\left((C, \Lambda_C)-\sum_{i=1} ^n(C_i, \Lambda_{C_i})\right)
-\sum_{F\in \calp} \phi(c(F))\chi(F, \Lambda _ F).
\mlabel{eq:chirec}
\end{equation}

Let $\mathcal{P}([n])$ denote the power set of $[n]$. Consider the surjective map
$$ \lambda: \mathcal{P}([n])\backslash \emptyset \longrightarrow \calt,\quad I\mapsto C_I.$$
For $H\in \calt$, denote
$$\lambda_H:=\sum\limits_{J\in \lambda^{-1}(H)} (-1)^{|J|-1}.$$

Then the \ValP of $\phi$ in Eq.~(\mref{eq:VaSubP}) can be expressed as
\begin{equation}
\phi(C, \Lambda_C)=\sum _{H\in \calt}\lambda_{\jh }\phi (H, \Lambda_H).
\mlabel{eq:phieq}
\end{equation}
Likewise, for $H\in \calp$ and the subdivision $\{t(C_i,H)\ | \ i\in J(H) \}$ of $t(C,H)$ in Lemma~\mref{lem:transversepartition}.(\mref{it:trpa}), the \ValP for this subdivision is
\begin{eqnarray}
\phi(t(C,H), \Lambda _{t(C,H)})&=&\sum _{I\subset J(H) }(-1) ^{|I|-1}\phi \left(\bigcap _{i\in I} t(C_i, H), \Lambda _{t(C,H)}\cap {\rm lin}(\bigcap _{i\in I} t(C_i, H))\right)
\mlabel{eq:transsubd}\\
&=&
\sum_{I\subset J(H) }(-1) ^{|I|-1}\phi \left(t(C_I, H), \Lambda _{t(C,H)}\cap {\rm lin}(t(C_I, H))\right)\notag
\end{eqnarray}
by Lemma \mref{lem:transversepartition}.(\mref{it:trpb}).
Furthermore, for the lattices on the right hand side, we have
$$\Lambda _{t(C,H)}\cap {\rm lin}(t(C_I, H))=\pi _{H^\perp}(\Lambda_C) \cap {\rm lin}\left(\pi _{H^\perp}(C_I)\right)=\pi _{H^\perp}(\Lambda_C) \cap \pi _{H^\perp}({\rm lin}(C_I))
$$
which agrees with
$$
\Lambda _{t(C_I,H)}=\pi _{H^\perp}(\Lambda_{C_I})=\pi _{H^\perp}(\Lambda_C \cap {\rm lin}(C_I))
$$
by Lemma \mref{lem:int}. Therefore, the \ValP in Eq.~(\mref{eq:transsubd}) becomes
\begin {equation}
\phi(t(C,H), \Lambda _{t(C,H)})=\sum _{I\subset J(H) }(-1) ^{|I|-1}\phi (t(C_I,H), \Lambda_{t(C_I,H)} )
=\sum _{F \in \calt(H)}\lambda_{\jf }\phi (t(F,H), \Lambda _{t(F,H)}), \mlabel{eq:phieqtr}
\end{equation}
where we have set
\begin{equation}\label{eq:calT} \calt(H):=\{C_I\,|\,I\subseteq J(H) \}.\end{equation}
Then with the notation in Eq.~(\ref{eq:degi}), we set
\begin{equation}\label{eq:caltl}\calt_\ell(H):=\{C_I\,|\,I\subseteq J(H) , j(C_I)=\ell\}.\end{equation}

Now we are ready to state the key combinatorial facts for the proof of Theorem~\mref{thm:IsubP}.

\begin{prop}
With  the above notations, the following equations hold for $\ell\geq 2$.  \begin{eqnarray}
&& \sum_{F\in \calp_\ell} \left (\phi(c(F))-\sum_{m=2}^{\ell-1} \sum_{H\in \calt_m(F)} \lambda_{\jh }\phi(t(H,F), \Lambda_{t(H,F)})\right) \chi(F, \Lambda_F)
\notag\\
&=& (\psi-\phi)\left(\sum_{F\in \calt_\ell} \lambda_{\jf }(F,\Lambda_F)\right) -\sum_{G\in\calp_\ci}\sum_{H\in \alpha^{-1}(G)_\ell}\phi(t(F,H), \Lambda_{t(F,H)})\chi(H, \Lambda_H) \mlabel{eq:phirec}\\
&&-\sum_{k\geq \ell+1}\sum_{H\in \calp_k}\sum_{F\in \calt_\ell(H)} \lambda_{\jf }\phi(t(F,H), \Lambda_{t(F,H)})\chi(H, \Lambda_H).
\notag \\
\chi((C,\Lambda_C)-\sum_i (C_i, \Lambda_{C_i}))&=&
(\psi-\phi)\left ( (C,\Lambda_C)-\sum_i (C_i, \Lambda_{C_i}) - \sum_{m=2}^{\ell-1}\sum_{F\in\calt_m} \lambda_{\jf }(F, \Lambda_F)\right) \notag\\
&& - \sum_{G\in \calp_\ci} \phi(t(C,G), \Lambda _{t(C,G)}) \left(\chi(G, \Lambda_G)-\sum_{m=1}^{\ell-1} \sum_{H\in \alpha^{-1}(G)_m}\chi(H, \Lambda_H)\right)\mlabel{eq:chiform}\\
&&-\sum_{k\geq \ell}\sum_{F\in \calp_k} \left(\phi(c(F))-\sum_{m=2}^{\ell-1}\sum_{H\in \calt_m(F)} \lambda_{\jh }\phi(t(H,F), \Lambda_{t(H,F)})\right) \chi(F, \Lambda_F). \notag
\end{eqnarray}
\mlabel{pp:chiform}
\end{prop}

\begin{proof} {\bf Proof of Eq.~(\mref{eq:phirec}): }
We have $\calp_\ell = \calt_\ell \coprod \calp'_\ell$ where $\calp'_\ell:=\calp_\ell \backslash \calt_\ell$.

We first consider the partial  sum
$$ \sum_{F\in \calt_\ell} \left (\phi(c(F))-\sum_{m=2}^{\ell-1} \sum_{H\in \calt_m(F)} \lambda_{\jh }\phi(t(H,F), \Lambda_{t(H,F)})\right) \chi(F, \Lambda_F)$$
on the left hand side of Eq.~(\mref{eq:phirec}) over the subset $\calt_\ell$ of $\calp_\ell$ introduced in (\ref{eq:caltl}). For $F\in \calt_\ell$, we have $F=\cap_{i\in J(F) } C_i$ with $j(F)=\ell$ and thus $\calt_\ell(F)=\{F\}$. Since $\ell\geq 2$, we have the disjoint union
$$\calt_\ell = (\calt_\ell\cap \calp_\ciii)\coprod (\calt_\ell\cap \calp_\civ) \coprod (\calt_\ell\cap \calp_\cv)$$
and
\begin{equation}
\calp'_\ell = (\calp'_\ell\cap \calp_\ciii)\coprod (\calp'_\ell\cap \calp_\civ) \coprod (\calp'_\ell\cap \calp_\cv)=\calp_\ciii'\coprod \calp_\civ'\coprod \calp_\cv'.
\label{eq:prime}
\end{equation}

Then by Eq.~(\mref{eq:phieqtr}) we have
\begin{eqnarray*}
&& \phi(c(F))-\sum_{m=2}^{\ell-1} \sum_{H\in\calt_m(F)} \lambda_{\jh }\phi(t(H,F), \Lambda_{t(H,F)}) \\
&=& \left\{ \begin{array}{ll}
\lambda_{\jf }\phi(t(F,F),\{0\})=\lambda_{\jf }, & \text{ for } F\in \calt_\ell\cap \calp_\ciii, \\
-\phi(t(C,F),\Lambda _{t(C,F)})+\lambda_{\jf }\phi(t(F,F), \{0\}) &\\
=-\phi(t(C,F),\Lambda _{t(C,F)})+\lambda_{\jf }, & \text{ for } F\in \calt_\ell\cap\calp_\civ, \\
-\phi(t(C,F),\Lambda _{t(C,F)})+\lambda_{\jf }\phi(t(F,F),\{0\}) =\lambda_{\jf }, & \text{ for } F\in \calt_\ell\cap\calp_\cv.
\end{array}
\right.
\end{eqnarray*}
where in the last case we have used $\phi(t(C,F),\Lambda _{t(C,F)})=0$ as a consequence of
Lemma~\ref{lem:transversepartition}.(\mref{it:trpc}).
Therefore we have

\begin{eqnarray}
&&\sum_{F\in \calt_\ell}\left(\phi(c(F))-\sum_{m=2}^{\ell-1} \sum_{H\in\calt_m(F)} \lambda_{\jh }\phi(t(H,F), \Lambda _{t(H,F)})\right)\chi(F, \Lambda_F) \mlabel{eq:rec1}\\
&=&
-\sum_{G\in \calp_\ci} \sum_{F\in \alpha^{-1}(G)_\ell\cap\calt} \phi(t(C,G),\Lambda _{t(C,G)})\chi(F, \Lambda_F)
+ \sum_{F\in \calt_\ell}\lambda_{\jf }\chi(F, \Lambda_F).
\notag
\end{eqnarray}

By definition, the second term on the right hand side of the above equation reads
$$(\psi-\phi)\left(\sum_{F\in \calt_\ell} \lambda_{\jf } (F, \Lambda_F)\right) - \sum_{(F,H)\in U_\ell} \lambda_{\jf }\phi(t(F,H),\Lambda _{t(F,H)})\chi(H,\Lambda_H),$$
where
$$U_\ell:=\left\{ (F,H)\,|\, F\in \calt_\ell, 0\neq H\precneqq F\right\}.$$
Note that $U_\ell$ is the disjoint union of the sets
$$U_{\ell,\ell}:=\{(F,H)\in U_\ell\,|\, j(H)=\ell\} \text{ and }
U_{\geq \ell +1}:=\{(F,H)\in U_\ell\,|\, j(H)\geq \ell+1\}.$$
On the one hand, for $(F,H)\in U_{\ell,\ell}$, we have $J(F) \subset J(H) $. The fact that they have the same cardinal implies the equality $J(F) = J(H) $. Moreover, $F=C_{J(F)}$. Since $H\precneqq F$, we have $H\in \calp'_\ell$ and obtain
$$ \sum_{(F,H)\in U_{\ell,\ell}} \lambda_{\jf }\phi(t(F,H),\Lambda _{t(F,H)})\chi(H,\Lambda_H) =\sum_{H\in \calp'_\ell}\sum_{F\in \calt_\ell(H)} \lambda_{\jf }\phi(t(F,H),\Lambda _{t(F,H)})\chi(H, \Lambda_H).$$

On the other hand,
\begin{eqnarray*}
\sum_{(F,H)\in U_{\ell,\geq\ell+1}} \lambda_{\jf }\phi(t(F,H),\Lambda _{t(F,H)})\chi(H, \Lambda_H)
&=& \sum_{0\neq H,j(H)\geq\ell+1}\sum_{F\in \calt_\ell(H)} \lambda_{\jf }\phi(t(F,H),\Lambda _{t(F,H)} )\chi(H,\Lambda_H)\\
&=& \sum_{k\geq \ell+1} \sum_{H\in \calp_k}\sum_{F\in \calt_\ell(H)} \lambda_{\jf }\phi(t(F,H),\Lambda _{t(F,H)})\chi(H,\Lambda_H).
\end{eqnarray*}

Inserting  the last two identities into Eq.~(\ref{eq:rec1}) yields the following expression for the left hand side of Eq.~(\mref{eq:phirec})
\begin{eqnarray*}
&&\sum_{F\in \calt_\ell \coprod \calp'_\ell}
\left (\phi(c(F))-\sum_{m=2}^{\ell-1} \sum_{H\in \calt_m(F)} \lambda_{\jh }\phi(t(H,F),\Lambda _{t(H,F)})\right) \chi(F, \Lambda_F) \\
&=&\sum_{F\in \calp'_\ell}
\left (\phi(c(F))-\sum_{m=2}^{\ell-1} \sum_{H\in \calt_m(F)} \lambda_{\jh }\phi(t(H,F),\Lambda _{t(H,F)})\right) \chi(F, \Lambda_F)\\
&&-\sum_{G\in \calp_\ci} \sum_{F\in \alpha^{-1}(G)_\ell\cap\calt} \phi(t(C,G),\Lambda _{t(C,G)})\chi(F)
+\sum_{F\in \calt_\ell}\lambda_{\jf }\chi(F, \Lambda_F)\\
&=& \sum_{F\in \calp'_\ell}
\left (\phi(c(F))-\sum_{m=2}^{\ell} \sum_{H\in \calt_m(F)} \lambda_{\jh }\phi(t(H,F),\Lambda _{t(H,F)})\right) \chi(F, \Lambda_F)\\
&&-\sum_{G\in \calp_\ci} \sum_{F\in \alpha^{-1}(G)_\ell\cap\calt} \phi(t(C,G),\Lambda _{t(C,G)})\chi(F, \Lambda_F)\\
&&+(\psi-\phi)\left(\sum_{F\in \calt_\ell} \lambda_{\jf } (F, \Lambda_F)\right)\\
&&-\sum_{k\geq \ell+1} \sum_{H\in \calp_k}\sum_{F\in \calt_\ell(H)} \lambda_{\jf }\phi(t(F,H),\Lambda _{t(H,F)})\chi(H, \Lambda_H).
\end{eqnarray*}

By Eq.~(\mref{eq:phieqtr}), the cofactor of $\chi(F,\Lambda_F)$ in the first sum in the above formula is
\begin{eqnarray*}
&&\phi(c(F))-\sum_{m=2}^{\ell} \sum_{H\in \calt_m(F)} \lambda_{\jh }\phi(t(H,F),\Lambda _{t(H,F)})\\
&=&
\left\{\begin{array}{ll}
0, & \text{ for } F\in \calp'_\ciii, \\
-\phi(t(C,F),\Lambda _{t(C,F)}), & \text{ for }F\in \calp'_\civ,\\
-\phi(t(C,F),\Lambda _{t(C,F)})=0, & \text{ for }F\in \calp'_\cv,
\end{array}\right.
\end{eqnarray*}
where we have applied the notations in Eq.~(\ref{eq:prime}), and in the last case, $\phi(t(C,F), \Lambda _{t(C,F)})=0$ using Lemma~\mref{lem:transversepartition}.(\mref{it:trpc}). Thus this sum
becomes
\begin{eqnarray*}
-\sum_{H\in \calp'_\civ} \phi(t(C,H),\Lambda _{t(C,H)})\chi(H, \Lambda_H)
&=&-\sum_{H\in \calp'_\civ} \phi(t(C,\alpha(H)),\Lambda _{t(C,\alpha (H))})\chi(H, \Lambda_H)\\
&=&-\sum_{G\in\calp_\ci}\sum_{H\in\alpha^{-1}(G)_\ell \cap \calp^\prime} \phi(t(C,G),\Lambda _{t(C,G)})\chi(H, \Lambda_H).
\end{eqnarray*}
This proves that the left hand side of Eq.~(\mref{eq:phirec}) agrees with the right hand side.
\smallskip

\noindent
{\bf Proof of Eq.~(\mref{eq:chiform}): } We prove the equation by induction on $\ell\geq 2$. We first verify the case when $\ell=2$. By definition,
\begin{eqnarray*}
\chi((C, \Lambda_C )-\sum_i (C_i, \Lambda_{C_i} ))&=& (\psi-\phi)\left((C, \Lambda_C )-\sum_i (C_i, \Lambda_{C_i} )\right)
-\sum_{F\in \calp} \phi(c(F))\chi(F, \Lambda_F) \\
&=& (\psi-\phi)\left((C, \Lambda_C )-\sum_i (C_i, \Lambda_{C_i} )\right)\\
&&-\sum_{F\in \calp_0\cup \calp_1} \phi(c(F))\chi(F, \Lambda_F)
-\sum_{k\geq 2}\sum_{F\in \calp_{k}} \phi(c(F))\chi(F, \Lambda_F).
\end{eqnarray*}
Now we see that the first and third sums on the right hand side readily agree with the corresponding sums on the right hand side of Eq.~(\mref{eq:chiform}).

For the second sum, note that
$$ \calp_0 = \calp_\ci, \quad
\calp_1=\calp_\cii \coprod \calp_{\ciii,1} \coprod \calp_{\civ,1}\coprod \calp_{\cv,1}.$$
By Lemma~\mref{lem:transversepartition}.(\mref{it:trpa}), $c(F)=0$ for $F\in \calp_\cii$. Also $\calp_{\ciii,1}=\emptyset$ by definition.
By Lemma~\mref{lem:transversepartition}.(\mref{it:trpb}), we have $\phi(t(C,F))=0$ for $F\in \calp_{\cv,1}$.

Notice that for $G\in \calp_{1}$, if $F^G\in \alpha ^{-1}(G)$ and $j(F^G)=1$, then we can take $J(F^G)=\{i\}$. So by Lemma \mref {lem:transversepartition}, $t(C,F)=t(C_i, F^G)$. This proves that the second sum agrees with  the second sum in Eq.~(\mref{eq:chiform}) when $\ell=2$. Therefore Eq.~(\mref{eq:chiform}) holds when $\ell=2$.

The inductive step follows from Eq.~(\mref{eq:phirec}) applied to the third sum.
\end{proof}

Now we are ready to prove Theorem~\mref{thm:IsubP}.
\smallskip

\noindent
\begin {proof} ({\it of\ \,Theorem~ \mref {thm:IsubP}}) We prove the statement by induction on the dimension of $C$,  the case   $\dim C=1$ being trivial.
Assume that the theorem holds for cones of dimension less or equal to $k\geq 1$ and consider a cone $C$ of dimension $k+1$. Let a subdivision of $C$ be given. Taking $\ell$ sufficiently large (say greater than the number $n$ of the $C_i$'s in the subdivision of $C$) in Eq.~(\mref {eq:chiform}), we have
\begin {eqnarray*}
\chi ((C, \Lambda_C )-\sum_i (C_i, \Lambda_{C_i} ))&=&(\psi -\phi)\left((C, \Lambda_C )-\sum_i (C_i, \Lambda_{C_i} )-\sum _{F\in \calt_{\geq 2}} \lambda_{\jf }(F,\Lambda_F) \right)\\
&&-\sum _{G\in \calp_\ci} \phi (t(C,G),\Lambda_{t(C,G)})\left(\chi (G, \Lambda_G)-\sum _{H\in \alpha ^{-1}(G)}\chi (H, \Lambda_H)\right).
\end{eqnarray*}
By the \ValP of $\phi$ and $\psi$, the first term on the right hand side is zero.
By the induction hypothesis, the second term is also zero since $\{H\in \alpha^{-1}(G)\}$ gives a subdivision of $G$ by Lemma~\mref{lem:alpha}.(\mref{it:alpha2}). This completes the induction.
\end{proof}

\section{Euler-Maclaurin formulae for \ltcones}
\mlabel{sec:emf}

We derive the Euler-Maclaurin formula from the above results combined with an \abf on lattice cones, which  generalizes Connes-Kreimer renormalization scheme.

\subsection{Meromorphicity of generating functions}
\mlabel{ss:reg}
From now on, we work in the filtered lattice space $\RR ^\infty$, with the standard lattice $\Lambda _\infty=\ZZ^\infty$ and a fixed basis $\{e_1, e_2, \cdots \}$.

To a cone $C$ in a lattice filtered space $V$, one can assign two meromorphic functions: the generating function (or the exponential discrete sum) $S(C)$ and the exponential integral $I_V(C)$ \cite {BV1,GP,GPZ2, La}. These can be extended to a {\ltcone} by the subdivision technique.

It is simple for simplicial cones. If $\lC\in \cc_k$ is a simplicial \ltcone (so in particular it is strongly convex), then the set
$$ \check C^-:=\check C^-_k:=\Big\{\vec \e:=\sum_{i=1}^k \e_ie_i^* \,\Big|\, \langle \vec x, \vec \e \rangle <0 \text{ for all } \vec x\in C\Big\}$$
is of dimension $k$. Here $\langle \vec x, \vec \e\rangle$ denotes the natural pairing $V_k\otimes V_k^\ast \to \RR$. Let $C^o$ denote the interior of $C$. For $\vec \e\in \check C^-$, then define
\begin{equation}
S^o\lC(\vec \e ): =
\sum_{\vec{n}\in C^o\cap \Lambda_C} e^{\langle \vec n, \vec \e \rangle}.
\mlabel{eq:osum}
\end{equation}

If $v_1, \cdots v_k \in \Lambda _C$ is a set of primary generators of $C$, and $u_1, \cdots, u_k$ is a basis of $\Lambda _C$, for $1\leq i\leq k$, let $v_i=\sum\limits_{j=1}^k a_{ji}u_j, a_{ji}\in \ZZ $. Define linear functions
$L_i:=L_{v_i}:=\sum\limits_{j=1}^k a_{ji}\langle u_j, \vec \e\rangle$ and let $w\lC $ denote the absolute value of the determinant of the matrix $[a_{ij}]$, then
\begin{equation}
I \lC (\vec \e ): =(-1)^k\frac
{w\lC}{L_1\cdots L_k}.
\mlabel{eq:ExpI}
\end{equation}

\begin{remark} We use a sign convention that is different from~\cite{GPZ2} in order to make the Euler-Maclaurin formula simpler.
\end{remark}

Then by the subdivision technique, we have
\begin{prop-def}
For a \ltcone \ $\lC$, the germ of functions $\sum\limits _{F\in \mathcal{F}^o(\conefamilyc )}   S^o(F,\Lambda_F)$ and $\sum\limits_{i\in [n]} I(C_i, \Lambda _{C_i})$ do not depend on the choice of the simplicial subdivision $\conefamilyc=\{(C_i,\Lambda _{C_i})\}_{i\in [n]}$ of $\lC$. Thus we define
$$S^o\lC :=\sum\limits _{F\in \mathcal{F}^o(\conefamilyc )}   S^o(F,\Lambda_F)
$$
and
$$I\lC =\sum_{i\in [n]} I(C_i, \Lambda _{C_i})
$$
for any simplicial subdivision $\conefamilyc=\{(C_i,\Lambda _{C_i})\}_{i\in [n]}$ of $\lC$.
\mlabel{defn:SI}
\end{prop-def}

We next view  the generating functions $S^o\lC(\vec \e)$   as  meromorphic germs with linear poles at zero, see \cite {GPZ3} for a more detailed discussion.

\begin{defn} Let $k$ be a positive integer.
\begin{enumerate}
\item
A {\bf germ of meromorphic functions at 0} on $\C^k$ is the quotient of two holomorphic functions in a neighborhood of 0 inside $\C^k$.
\item
A germ of meromorphic functions $f(\vec \e)$ on $\C^k$ is said to have {\bf  linear poles at zero with lattice coefficients} if there exist vectors $L_1, \cdots, L_n\in \Lambda_k\otimes \QQ  $ (possibly with repetitions) such that $f\,\Pi_{i=1}^n L_i$ is a holomorphic germ at zero whose Taylor expansion has lattice coefficients.
\item
We will denote by $\calm_\QQ  (\C^k)$  the set of germs of meromorphic functions on $\C^k$ with linear poles at zero with lattice coefficients. It is a linear subspace over $\QQ$.
\end{enumerate}
\mlabel{de:fr}
\end{defn}

Then composing with the projection $\C^{k+1} \to \C^k$   dual to the inclusion $j_k:\C^k\to \C^{k+1}$ yields the embedding
$$\calm_\coef  (\C^k)\hookrightarrow \calm_\coef  (\C^{k+1}),$$
thus giving rise to the direct limit
$$\calm_\coef  ( \C^\infty):
=\dirlim  \calm_\coef  (\C^k)
=\bigcup_{k=1}^\infty  \calm_\coef  (\C^k).
$$

\begin{lemma} For a simplicial \ltcone \ $\lC\in \cc _k$, the germs of functions $S^o\lC(\vec \e )$ lies in $\calm_\Q(\C ^k)$.
\mlabel{lem:rationality}
\end{lemma}

\begin{proof} We first prove the proposition for a smooth {\ltcone}\  $\lC$.
Let $C= \langle v_1,\cdots, v_m\rangle$ with $\{v_1,\cdots, v_m\}$ being a basis of $\Lambda _C$.
Since an element $\vec x$ in $C\cap \Lambda _C$ can be written in a unique way as $\sum\limits_{j=1}^m n_j v_j$ where $n_j\in \Z_{\geq 0}$,  for $\vec \e =\sum\limits_{j=1}^m \e_je_j^* \in  \check C^-$, we have
\begin{equation}
S^o\lC (\vec \e )
 :=   \prod_{j=1}^m \sum_{ n_j\in \Z_{\geq 1}} e^{  n_j\, \langle v_j, \vec \e \rangle }
 =   \prod_{j=1}^m\frac{e^{\langle v_j,\vec \e\rangle}}{1-e^{\langle v_j, \vec \e\rangle}}
 =   \prod_{j=1}^m\frac{e^{L_j(\vec \e )}}{1-e^{L_j(\vec \e )} },
\mlabel{eq:SOpenSmooth}\end{equation}
where $L_j(\vec \e )=\langle v_j, \vec \e \rangle$. They are holomorphic on $\check C^{-}$ and extend to germs of meromorphic functions on $\C^k$ with simple linear poles at $L_1(\vec \e )=0, \cdots ,L_{n}(\vec \e )=0$.

Indeed, from the generating power series $\frac{x}{e^x-1}=\sum\limits_{n=0}^\infty B_n \frac{x^n}{n!}$ of Bernoulli numbers, we have that $\frac{1}{1-e^x}=-\frac{1}{x} \frac{x}{e^x-1}$ is in $\calm_\Q(\C)$. Then the same holds for $\frac{e^x}{1-e^x}=\frac{1}{1-e^x}-1$. Thus for each linear form $L$ on $\C ^k$ with lattice coefficients, both $\frac{L}{1-e^L}$ and $\frac{e^L}{1-e^L}$ are in $\calm_\Q(\C ^k)$. For a smooth \ltcone, the conclusion that $S^o\lC(\vec \e )$ lies in $\calm_\Q(\C ^k)$  follows from Eq.~(\mref{eq:SOpenSmooth}) since $\calm_\Q(\C ^k)$ is closed under multiplication.

Next for a simplicial \ltcone \ $\lC$, we prove the statement by taking a smooth subdivision and applying Proposition-Definition~\mref{defn:SI}, noting that faces of a smooth \ltcone are smooth by Proposition~\mref{pp:smface}.
\end{proof}

Therefore, we have linear map
$$S^o: \QQ\cc \to \calm_\Q(\C^\infty), \quad \lC\mapsto S^o\lC.$$
By definition, the following conclusion holds.
\begin{coro}Let $\lC$ be a \ltcone \ and let $\conefamilyc =\{(C_1, \Lambda _C),\cdots, (C_r, \Lambda _C) \}$ be  a (not necessarily simplicial)  subdivision of $C$.
Then we have
$$S^o\lC=\sum _{F\in \mathcal{F}^o(\conefamilyc)}   S^o(F,\Lambda_F)$$
and
$$I\lC =\sum_{i\in [n]} I(C_i, \Lambda _{C_i})$$
in $\calm_\Q(\C^\infty)$, that is, $S^o$ has the \SSubP and $I$ has the \ISubP.
\mlabel{coro:MlSub}
\end{coro}

\begin{defn}
For a \ltcone $\lC\in \cc _k$, define its {\bf (closed) generating function} by
\begin{equation}S^c \lC= \sum_{F\preceq C} S^o (F,\Lambda _F),
\mlabel{eq:rcczv}
\end{equation}
giving rise to the linear map
$$S^c: \QQ\cc \to \calm_\Q(\C^\infty),\quad \lC\mapsto S^c\lC.$$
\end{defn}

By Proposition \mref {prop:ValPSSubP},
we have
\begin {coro} $S^c:\QQ \cc \to \calm_\Q(\C ^\infty)$ has the \ValP.
\end{coro}

We now state one more key property of $S^o\lC$ and $S^c\lC$.
\begin {prop} If $\lC$ is not strictly convex, then $S^o\lC $ and $S^c \lC$ are both zero.
\mlabel{pp:0S}
\end{prop}

\begin{proof}
First consider the case when $C$ is a one-dimensional subspace.
So $(C,\Lambda_C)=(\RR_{\geq 0} u, \ZZ u)$. Then $\{\cone{u}, \cone{-u}\}$ is a smooth subdivision of $C$. Then as in Eq.~(\mref{eq:SOpenSmooth}), we obtain
$$S^o (C,\ZZ u)(\vec \e )=S^o(\cone{u},\ZZ u)(\vec \e)+S^o(\{0\},\{0\})(\vec \e) +S^o(\cone{-u},\ZZ u)(\vec \e)
=\frac {e^{<u,\vec \e >}}{1-e^{<u,\vec \e>}}+1+\frac {e^{<-u,\vec \e >}}{1-e^{<-u,\vec \e>}}=0.
$$
Since $C=\RR u$ does not have a proper face, by Eq.~(\mref{eq:rcczv}) we have
$$S^c (C,\ZZ u)(\vec \e )=S^o (C,\ZZ u)(\vec \e )=0.$$

Next consider the case  $(C,\Lambda _C)$ where $C$ is a linear space of dimension $k$. Then $C$ has no proper face. Take a lattice basis $\{v_1, \cdots , v_k\}$ of $\Lambda _C$ and denote $C_{\alpha_1\alpha_2\cdots\alpha_k}: =\cone{\alpha_1 v_1,\alpha_2 v_2,\cdots ,\alpha_k v_k}$ for $\alpha_i\in \RR, 1\leq i\leq k$.
Then the family of {\ltcone}s $\{(C_{\alpha_1\alpha_2\cdots\alpha_k}, \Lambda _C)\ |\, \alpha _i=\pm 1, 1\leq i\leq k\}$ provides  a simplicial subdivision of $\lC$. Thus
${\mathcal F}^o(C\sim \cup C_{\alpha_1\alpha_2\cdots\alpha_k})=\{C_{\alpha_1\alpha_2\cdots\alpha_k} \,|\, \alpha _i=0, \pm 1, 1\leq i\leq k\}
$
and
$$S^o(C_{\alpha_1\alpha_2\cdots\alpha_k}, \Lambda _C\cap \lin( C_{\alpha_1\alpha_2\cdots\alpha_k} )) (\vec \e)=\prod _{i, \alpha _i\not =0}\frac{e^{<\alpha_i v_i, \vec \e >}}{1-e^{<\alpha_i v_i,\vec  \e >}}.
$$
Thus
\begin{eqnarray*}
S^o\lC(\vec \e)&=&\sum_{\alpha_i=0,\pm 1, 1\leq i\leq k} S^o(C_{\alpha_1\alpha_2\cdots\alpha_k}, \Lambda _C\cap \lin(C_{\alpha_1\alpha_2\cdots\alpha_k} )) (\vec \e) \\
&=&\prod _i\left(\frac {e^{<v_i,\vec \e >}}{1-e^{<v_i,\vec \e>}}+1+\frac {e^{<-v_i,\vec \e >}}{1-e^{<-v_i,\vec \e>}}\right)=0.
\end{eqnarray*}

Finally consider the case when $C$ is a cone that contains a linear subspace. By Proposition 3.4.(a) in \cite {GPZ2}, we have
$C=\{v+u\ | \ v \in L, u\in C'\}$, where $L$ is a linear subspace and $C'$ is a strongly convex cone in the orthogonal complement $\lin (L;\lin (C))^\perp$ of $L$ in $\lin(C)$. Therefore any element in $C$ has a unique decomposition $v+u$ with $v\in L$ and $u \in C'$. Let $\Lambda _L $ and $\Lambda _{C'}$ be the projection of $\Lambda _C$ in $L$ and $\lin^\perp (L; {\rm lin} (C))$ respectively.
Picking a basis $\{v_1, \cdots , v_k\}$ of $\Lambda _L$,
the set $\{C_{\alpha_1,\alpha_2,\cdots,\alpha_k}+ C^\prime\,|\,\alpha _i=\pm 1\}$ provides  a subdivision of $C$. Further,
$$S^o(C_{\alpha_1,\alpha_2,\cdots,\alpha_k}+C^\prime, \Lambda_C )(\vec \e)=S^o(C_{\alpha_1,\alpha_2,\cdots,\alpha_k}, \Lambda_L)(\vec \e)S^o(C^\prime,\Lambda_{C^\prime})(\vec \e).
$$
So as in the case of a linear subspace, we have
$$S^o \lC(\vec \e)=S^o(L, \Lambda_L)(\vec \e)S^o(C', \Lambda _{C'})(\vec \e)=0.
$$

For $S^c\lC$, note that any face of $C$ contains the above $L$. Therefore $S^c\lC=0$ by Eq.~(\mref{eq:rcczv}).
\end {proof}

\subsection{\abf}
We first give a general formulation of the \abf before applying it to the study of lattice cones.
\subsubsection{The general result}
\mlabel{sec:abd}
We give a generalization of the \abf of Connes-Kreimer~\mcite{CK} for connected coalgebras without the need for either a Hopf algebra or a Rota-Baxter algebra. We begin with a lemma (see e.g. \cite[Prop. II.3.1]{Ma}).

\begin{lemma}
Let $\coalg=\bigoplus_{n\geq 0} \coalg^{(n)}$ be a connected cograded coaugmented coalgebra  with coaumentation $u$. Denote $J=u(1)$ and let $A$ be a commutative algebra with unit $1_A$. Let $\ast$ be the {\bf convolution product} on the algebra ${\mathcal L}(\coalg,A)$ of linear maps from $\coalg $ to $A$ and let $\varphi\in {\mathcal L}(\coalg,A)$ be  such that  $\varphi(\I )=1_A$. Then $\varphi$ has a convolution inverse $\varphi^{\ast (-1)}:\coalg\to A$ for which $\varphi^{\ast (-1)}(\I )=1_A$. Consequently, $${\mathcal G}(\coalg,A):=\{\varphi\in {\mathcal L}(\coalg,A)\,\big|\,  \varphi(\I )=1_A\}$$ endowed with the convolution product is a group.
\mlabel{lem:phiinv}
\end{lemma}

\begin{thm}
Let $\coalg=\bigoplus_{n\geq 0} \coalg^{(n)}$ be a connected cograded coaugmented coalgebra. Let $A$ be a unitary algebra. Let $A=A_1\oplus A_2$ be a linear decomposition  such that $1_A\in A_1$.
 Let $P$ be the projection of $A$ to $A_1$ along $A_2$.
Given  $\varphi\in {\mathcal G}(\coalg,A)$,
define maps $\varphi_i\in {\mathcal G}(\coalg,A), i=1,2$,  by the following recursive formulae on $\ker \e$:
\begin{eqnarray}
\varphi_1(x)&=&-P\Big(\varphi(x)+\sum_{(x)} \varphi_1(x')\varphi(x'')\Big),  \mlabel{eq:phi-}\\
\varphi_2(x)&=&(\id_A-P)\Big(\varphi(x)+\sum_{(x)} \varphi_1(x')\varphi(x'')\Big).
\mlabel{eq:phi+}
\end{eqnarray}
\begin{enumerate}
\item
We have  $\varphi_i(\ker \e)\subseteq A_i$  $($hence  $\varphi_i: \coalg \to \bfk 1_A + A_i$$)$. Moreover, the following factorization holds
\begin{equation}
\varphi=\varphi_1^{\ast (-1)} \ast \varphi_2.
\mlabel{eq:abf}
\end{equation}
\mlabel{it:abf}
\item
$\varphi_1$ and $\varphi_2$ are the unique maps in $ {\mathcal G}(\coalg,A)$  such that   $ \varphi_i(\ker \e)\subseteq A_i$ for $i=1, 2,$ and Eq.~(\mref{eq:abf}) holds.
\mlabel{it:uniq}
\item If moreover $A_1$ is a subalgebra of $A$, then $\phi_1 ^{\ast (-1)}$ lies in $ {\mathcal G}(\coalg,A_1)$. \mlabel{it:phiinv}
\end{enumerate}
\mlabel{thm:abf}
\mlabel{thm:Birkhoff}
\end{thm}

\begin{proof}
(\mref{it:abf})
The inclusion $\varphi_i(\ker \e)\subseteq A_i, i=1,2,$  follows from the definitions. Further
$$\varphi_2(x)= (\id_A -P)\Big(\varphi(x)+\sum_{(x)}\varphi_1(x')\varphi(x'')\Big)
= \varphi(x) +\varphi_1(x) +\sum_{(x)}\varphi_1(x')\varphi(x'')
= (\varphi_1\ast \varphi)(x).
$$
Since $\varphi_1(\I)= 1_A$, $\varphi_1$ is invertible for the  convolution product in $A$ by Lemma~\mref{lem:phiinv}. Then Eq.~(\mref{eq:abf}) follows.

\smallskip

\noindent
(\mref{it:uniq}) Suppose there are $\psi_i\in {\mathcal G}\left(\coalg,A\right), i=1,2,$ with $  \psi_i(\ker \e)\subseteq A_i$ such that $\varphi=\psi_1^{\ast (-1)} \ast \psi_2 $.
We prove $\varphi_i(x)=\psi_i(x)$ for $i=1,2, x\in \coalg^{(k)}$ by induction on $k\geq 0$. These equations hold for $k=0$. Assume that the equations hold for $x\in\coalg^{(k)}$ where $k\geq 0$. For $x\in \coalg^{(k+1)}\subseteq \ker(\e)$, by $  \varphi_2=\varphi_1\ast \varphi$ and  $  \psi_2= \psi_1 \ast \varphi,$  we have
 $$ \varphi_{2}(x)= \varphi_1(x)+ \varphi(x)+ \sum_{(x)}\varphi_1(x^\prime)\varphi(x^{\prime \prime}), \quad \psi_{2}(x)= \psi_1(x)+ \varphi(x)+ \sum_{(x)}\psi_1(x^\prime)\varphi(x^{\prime \prime}),$$
where we have made use of
 $\varphi_1(\I)=\psi_1(\I)=\varphi(\I)=1_A$ . Hence by the induction hypothesis, we have
 $$
 \varphi_{2}(x)-\psi_{2}(x)=\varphi_{1}(x)-\psi_{1}(x)+\sum_{(x)}\big(\varphi_{1}(x^{\prime  })-\psi_{1} (x^{\prime  })\big) \varphi(x^{\prime \prime})=\varphi_{1}(x)-\psi_{1} (x)\in A_{1}\cap A_2=\{0\}.$$
 Thus
$\varphi_i(x)=\psi_i (x), i=1,2,$ for all $x\in \ker(\e)$.
\smallskip

\noindent (\mref{it:phiinv}) If   $A_1$ is a subalgebra, then it follows from Lemma \mref{lem:phiinv} applied to $A_1$ instead of $A$, that
$\varphi_1$ is invertible in $A_1$.
\end{proof}

\subsubsection{Application to lattice cones}

We now focus on the filtered lattice space $\RR ^\infty$, let  $Q(\cdot,\cdot)$ denote the inner product chosen in Eq.~(\ref{eq:Q}).
In this setup, we have constructed two linear maps:
$$S^o: \QQ \cc\to \calm_\Q(\C ^\infty)
\quad \text{and}
\quad S^c: \QQ \cc\to \calm_\Q(\C ^\infty).
$$

Let $\calm_{\coef ,+}(\C^k)$ denote the space of  germs of holomorphic functions at zero in $\C^k$ whose Taylor expansions at zero have lattice coefficients. We set
$$\calm_{\coef ,+} (\C^\infty) :=\dirlim \calm_{\coef ,+} (\C^k) =\bigcup_{k=0}^\infty\calm_{\coef ,+} (\C^k).$$
Then $\calm_{\coef ,+} (\C^k)$ (resp. $\calm_{\coef ,+} (\C^\infty)$) is a unitary subalgebra of $\calm _{\QQ }(\C ^k)$ (resp. $\calm _{\QQ }(\C ^\infty)$).

The filtered lattice Euclidean space $(\RR ^\infty, Q(\cdot, \cdot))$ allows us to apply \cite[Theorem~4.4]{GPZ3} to obtain the linear decomposition
\begin{equation}\label{eq:decocalm}\calm _{\QQ }(\C ^\infty)=\calm _{\QQ, + }(\C ^\infty)\oplus \calm _{\QQ, -}(\C ^\infty).\end{equation}
Here $\calm_{\QQ,-}(\C^\infty)$ is the subspace spanned by {\bf polar germs}, defined to be
germs of meromorphic functions at zero of the form
$$\frac {h(\ell_1, \cdots , \ell_m)}{L_1^{s_1}\cdots L_n^{s_n}},
$$
where
\begin{enumerate}
\item
$h$ lies in $\calm_{\coef ,+}(\C^m)$,
\item
$\ell_1, \cdots, \ell_m, L_1, \cdots ,L_n$  lie in $\Lambda _k\otimes \coef $, with $L_1, \cdots ,L_n$ linearly independent, such that
$$Q(\ell_i, L_j)=0 \quad \text{ for all } (i,j)\in [m]\times [n].$$
\item
$s_1,\cdots, s_n$ are positive integers.
\end{enumerate}

Since $\calm _{\Q,+}(\C ^\infty )$ is a unitary subalgebra, the \abf in Theorem~\mref{thm:abf} applies, with $\coalg=\QQ \cc$ and
\begin{equation}\label{eq:pi+}A= \calm_\Q (\C ^\infty),\quad A_1= \calm _{\Q,+}(\C ^\infty), \quad A_2=\calm_{\Q,-}(\C^\infty),\quad P=\pi _+:\calm_\QQ(\C ^\infty)\to \calm_{\QQ,+}(\C ^\infty),
\end{equation}
which is the orthogonal projection onto the holomorphic part $\calm_{\QQ,+}(\C ^\infty)$ along the space   $\calm_{\QQ,-}(\C ^\infty)$ of polar germs by means of the decomposition in Eq.~(\ref{eq:decocalm}).
We consequently obtain the following theorem.
\begin {coro} $($\text{\bf{\abf for generating functions}}$)$ For the linear map
$$S^\#:\QQ \cc\to \calm_\Q(\C ^\infty),
$$
where $\#\in \{o, c\}$, there exist unique linear maps $S^\#_1: \QQ \cc \to \calm_{\Q,+}(\C^\infty)$ and
$S^\#_2: \QQ \cc \to \Q+\calm ^R_{\Q,-}(\C^\infty)$, with $S^\#_1(\{0\},\{0\})=1$, $S^\#_2(\{0\},\{0\})=1$, such that
 \begin{equation}
 S^\#= (S^\#_1)^{\ast (-1)}\ast S^\#_2.
\mlabel{eq:abfdC}
\end{equation}
\mlabel {cor:abfd}
\end{coro}

We shall provide an interpretation of the linear maps $S^o_1$, $S^o_2$ and $S^c_1$, $S^c_2$ in the context of Euler-Maclaurin formula. Before that,  we deduce a relation between the factors in open and closed cases.
\begin{prop} For $\lC \in \cc$, we have
\begin{equation}
S^o_2\lC= S^c_2 \lC
\mlabel{eq:cumfo2}
\end{equation}
and
\begin{equation}
(S^c_1)^{*(-1)}\lC=\sum _{G \preceq C}(S^o_1)^{*(-1)} (G, \Lambda_G)
\mlabel{eq:cumfo1}
\end{equation}
\mlabel{pp:oeml}
\end{prop}

\begin{proof} Let $\lC\in \cc$. By Eqs.~(\mref{eq:rcczv}) and (\mref{eq:abf}) we have
\begin{eqnarray*}
S^c \lC&=&\sum _{F\preceq C}S^o(F,\Lambda _F)\\
&=&\sum _{F\preceq C} \sum_{G\preceq F}(S^o_1)^{*(-1)} (t(F,G), \Lambda_{t(F,G)})S^o_2(G,\Lambda _G)\\
&=&\sum _{G\preceq C}\left(\sum _{G\preceq F\preceq C} (S^o_1)^{*(-1)} (t(F,G),\Lambda _{t(F,G)})\right)S^o_2(G,\Lambda _G).
\end{eqnarray*}
Let $B(C,\Lambda_C)$ denote the right hand side of Eq.~(\mref{eq:cumfo1}). Then we have
$$B(t(C,G),\Lambda_{t(C,G)})=\sum_{H\preceq t(C,G)} (S^0_1)^{\ast (-1)}(H,\Lambda_H) =\sum_{G\preceq F\preceq C}(S^o_1)^{\ast (-1)}(t(F,G),\Lambda_{t(F,G)})$$
by Proposition~\mref{pp:transversecone}.(\mref{it:com1d}). Thus
$ S^c \lC = B\ast S^o_2$. Since the ranges of $B$ and $S^o_2$ are in $\calm_{\QQ,+}$ and $\calm_{\QQ,-}$ respectively, the desired equations hold from the uniqueness of the \abf $S^c=(S^c_1)^{\ast (-1)}\ast S^c_2$.
\end{proof}

\subsubsection {Continuous subdivision property of $S^c_2$}
\mlabel{ss:subds}
Now let us study the subdivision properties of factors in the \abfs for $S^c$ and $S^c$. Let $$\mu ^o=(S^o_1)^{*(-1)}, \quad \mu ^c=(S^c_1)^{*(-1)}.$$

\begin{theorem} The  linear maps $\mu^o$ and $\mu^c$ on $\QQ \cc$ have the projection formulae:
  $$\mu^o =\pi_+ \,S^o\quad \text{and} \quad
\mu^c =\pi_+ \,S^c .
  $$
\mlabel{prop:mu}
\end{theorem}
\begin{proof} Let $\lC$ be a \ltcone \ and $(F,\Lambda_F)$ be a face of it. Since the linear spaces $\lin F$ and  $\lin t(C,F)$ are perpendicular in $V$ for the inner product $Q(\cdot,\cdot)$, the linear functions from $F$ and from $t(C,F)$ are perpendicular. Thus, for $F\neq \{0\}$ we have
$$\pi _+\left((S_1^o)^{*(-1)}(t(C,F), \Lambda _{t(C,F)})({\rm id}-\pi_+)(S^o_2(F,\Lambda_F))\right )=0.
$$
It then follows from  Theorem~\mref{thm:abf} that, for $\lC\not=\left(\{0\},\{0\}\right)$
\begin{eqnarray*}
\lefteqn{S^o_1\lC=-\pi_+\left(S^o\lC +\sum _{\{0\}\precneqq F\precneqq C}S_1^o(t(C,F),\Lambda_{t(C,F)})S^o(F, \Lambda_F)\right)}\\
&=&-\pi_+\left(S^o\lC +\sum _{\{0\}\precneqq F\precneqq C}S_1^o(t(C,F), \Lambda_{t(C,F)})\pi_+(S^o(F, \Lambda_F))\right)\\
&=& -\pi_+(S^o\lC)-\sum _{\{0\}\precneqq F\precneqq C}S_1^o(t(C,F),\Lambda_{t(C,F)})\pi_+(S^o(F,\Lambda_F)).
\end{eqnarray*}
Hence,
$$-\left(S^o_1\ast(\pi_+\,S^o)\right)\lC=0 \text{ for }\lC\neq (\{0\},\{0\}),
$$
which combined with $(S^o_1\ast (\pi_+\,S^o))(\{0\},\{0\})=1$ yields
$$\mu^o = (S^o_1)^{*(-1)}=\pi_+\,S^o.
$$

The same proof yields the corresponding formula for the closed case.
\end{proof}

We have the following direct consequences.
\begin{coro}
\begin{enumerate}
\item
$\mu ^o$ has the \SSubP and $\mu ^c$ has  the \ValP.
\label{cor:musubd}
\item
$S^c_2$ and thus $S^o_2$ have the \ISubP.
\mlabel {thm:MuSub}
\end{enumerate}
\label{co:subd}
\end{coro}
\begin {proof}
(\ref{cor:musubd}) follows from Theorem~\ref{prop:mu} and the linearity of $\pi_+$. Then (\ref{thm:MuSub}) follows from Corollary~\ref{cor:abfd} and Theorem \mref {thm:0-sub}.
\end{proof}

\subsection{Euler-Maclaurin formula}
With all the preparations accumulated so far, we are ready to derive Euler-Maclaurin formulae for {\ltcone}s. We keep the setup in Section~\mref {ss:subds}. We identify the map $S^c_2$ for smooth {\ltcone}s first. Let $\pi_{\pm}$ be the projection of $\calm _{\QQ}(\C ^\infty)$ to $\calm _{\QQ, \pm}(\C ^\infty)$.

\begin {prop}
\mlabel {pp:SmoothMu} For a smooth \ltcone \ $\lC$, we have
$$\pi _-S^c \lC (\vec \e )=\sum _{F\precneqq C}  \pi _+S^c(t(C,F), \Lambda _{t(C,F)})(\vec \e )I(F, \Lambda _F)(\vec \e ).
$$
\end{prop}

\begin {proof} Let  $\lC$ be a  smooth \ltcone  with primary generators $v_1, \cdots, v_n\in \Lambda_C$, and  let $L_i=L_{v_i}, i=1,\cdots,n,$ be the linear function $\langle v_i, \vec \e \rangle$ where $\vec \e \in V^*$. Furthermore, let
$\displaystyle{\frac 1{1-e^x}=-\frac 1x+h(x)}$
be the decomposition of the function $\displaystyle{\frac 1{1-e^x}}$ into its singular part and holomorphic part.
Then
$$S^c\lC (\vec \e )=\prod_{i=1}^n \frac {1}{1-e^{L_{i}}}
=\prod_{i=1}^n (I(L_{i})+h(L_{i})), \quad I(L_i)=-\frac{1}{L_i}.
$$

For any nonempty subset $J\subset [n]$, let $L_J=\Pi _{i\in J}L_i$ and, as a convention, let $L_{\emptyset }=1$. then the denominators in the expansion are of the form $L_J$ for some nonempty $J\subset [n]$.

Note that any face $(F, \Lambda _F)$ of $(C,\Lambda_C)$ is of the form $\langle v_i\,|\, i\in J\rangle$ for $\emptyset \neq J\subseteq [n]$. Thus we focus on the polar germ of the form $L_Jg$ with $g$ a holomorphic function in linear forms perpendicular to $L_i , i\in J$. Notice in this case the \ltcone is smooth, thus $I(F, \Lambda _F)(\vec \e )=L_J$. By our decomposition procedure, terms of this type  come from the projection of
$$L_J\prod _{i\in [n]-J} \frac {1}{1-e^{L_i}}.
$$
For any vector $v$ and any non-empty subset $K\subset [n]$, let $v^K$ be the projection of $v $ to  the orthogonal subspace to the subspace spanned by $v_{j}, j\in K$, and $L_{v}^J$ be the corresponding linear function. We also use $v^{JK}$ and $L_v^{JK}$ to denote the vector $(v^J)^K$ and the corresponding linear function.  With these notations we have
$$L_i =L_i^K+\sum _{j\in K} a_{ij}L_{j}
$$
for some constants $a_{ij}$, $j\in K$.

The part of the form $L_Jg$ in $\pi _-S^c\lC$ coincides with the corresponding part in
$$\pi _-\left(L_J\prod _{i\in [n]-J} (I(L_i)+h(L_{i}^J))\right).
$$
Now let us determine the contribution from
$$L_JL_{K-J}\prod _{i\in [n]-K} h(L_{i}^J),
$$
with $J\subset K\subsetneqq [n]$.
By above notation, and for $\ell \in [n]-K$, we have
\begin{eqnarray*}
L_{\ell}^J&=&L_{\ell}^{JK}
+\sum _{i\in K} b_{\ell i}L_i\\
&=&L_{\ell}^{JK}+\sum _{i\in J} b_{\ell i}L_i+\sum _{i\in K-J} b_{\ell i}L_{i}^J+\sum _{i\in K_J} \sum _{j\in J} b_{\ell i}a_{ij}L_{j}\\
&=&L_{\ell}^{JK}+\sum _{i\in K-J} b_{\ell i}L_{i}^J+\sum _{i\in J} c_{\ell i}L_{i}.
\end{eqnarray*}
Since  the spaces spanned by $\{v_i\,|\, i\in K\}$ and by $\{v_i\,|\, i\in J\}\cup\{v_j^J\,|\, j\in K-J\}$ coincide, for any vector, its orthogonal projection to the space spanned by $\{v_i\,|\, i\in K\}$ and to the space spanned by by $\{v_i\,|\, i\in J\}\cup\{v_j^J\,|\, j\in K-J\}$ are the same.

Therefore,
$$L_{\ell}^J=L_{\ell}^{JK}+\sum _{i\in K-J} b_{\ell i}L_i+\sum _{i\in J} b_{\ell i}L_i
$$
is the the projections of $L_{\ell}^J$ with respect to the spaces spanned by $\{v_i\,|\, i\in K\}$,
and
$$L_{\ell}^J=L_{\ell}^{JK}+\sum _{i\in K-J} b_{\ell i}L_{i}^J+\sum _{i\in J} c_{\ell i}L_{i}
$$
is the the projections of $L_{\ell}^J$ with respect to the spaces spanned by $\{v_i\,|\, i\in J\}\cup\{v_j^J\,|\, j\in K-J\}$.

By direct calculations, the polar germs of type $L_Jg$ arising from
$L_J\prod\limits _{i\in K-J} I(L_{i})\prod\limits _{i\in [n]-K} (L_{i}^J)^{\alpha _i}$
coincide with those from
$L_J\prod\limits _{i\in K-J} I(L_{i}^J)\prod\limits _{i\in [n]-K} (L_{i}^J)^{\alpha _i},$
for fixed $a_i \in \ZZ _{\ge 0}, i\in K-J$.
Therefore, the polar germs of type $L_Jg$ arising from
$L_J\prod\limits _{i\in K-J} I(L_{i})\prod\limits _{i\in [n]-K} h(L_{i}^J)$
coincide with those from
$L_J\prod\limits _{i\in K-J} I(L_{i}^J)\prod\limits _{i\in [n]-K} h(L_{i}^J).
$
Consequently,  the corresponding terms in
$L_J\prod\limits _{i\in [n]-J} (I(L_i)+h(L_{i}^J))
$
coincide with those from
$L_J\prod \limits_{i\in [n]-J} (I(L_{i}^J+h(L_{i}^J)),
$
which is of the form
$I(L_1)\cdots I(L_m)\pi _+ S^c(t(C,F),\Lambda _{t(C,F)}).
$
This completes the proof.
\end{proof}

\begin {coro} Let $\lC$ be a smooth {\ltcone}, then
$$S^c_2\lC (\vec \e) =I\lC (\vec \e)
$$
and we have the {Euler-Maclaurin formula}:
$$S^c\lC =\sum _{F\preceq C} \pi _+S^c(t(C,F), \Lambda _{t(C,F)})I(F, \Lambda _F).
$$
\mlabel {coro:SmoothEMF}
\end{coro}
\begin{proof}
Adding $\pi_+S^c(C,\Lambda_C)$ to both sides of the equation in Proposition~\ref{pp:SmoothMu}, we obtain
$$ S^c(C,\Lambda_C)(\vec \e)=\sum _{F\preceq C}  \pi _+S^c(t(C,F), \Lambda _{t(C,F)})(\vec \e )I(F, \Lambda _F)(\vec \e ).
$$
Then the corollary follows from the uniqueness of the \abf in Corollary~\ref{cor:abfd} since $I(F,\Lambda_F)$ is in $\calm_{\Q,-}(\C^\infty)$.
\end{proof}

We are now ready to give the Euler-Maclaurin formula for {\ltcone}s. Recall that a  cone $C$ in a lattice vector space $(V, \Lambda _V )$ can be viewed as a \ltcone $(C, {\rm lin}(C)\cap \Lambda _V )$. Our approach by means of the \abf applied to  $(C, {\rm lin}(C)\cap \Lambda _V )$ yields  back Berline-Vergne's Euler-Maclaurin formulae for the cone $C$~\mcite{BV1} together with a new piece of information, namely that the interpolation function $\mu^c$ actually boils down to  the holomorphic projection of the exponential sum.

\begin {theorem}\label{thm:S1S2} Let $\lC\in \cc_k$ be a lattice cone and $\vec \e \in V_k^*$. Then
$$S^c_2\lC (\vec \e)=I\lC (\vec \e),$$
$$\hspace{-3cm}(\text{\bf Euler-Maclaurin formula})\qquad S^c\lC =\sum _{F\preceq C} \mu^c(t(C,F), \Lambda _{t(C,F)})I(F, \Lambda _F)
$$
and the interpolation function $\mu^c$ coincides with the holomorphic projection of the discrete sum,
$$\mu ^c =\pi _+S^c.$$
\end{theorem}

\begin{proof} Proposition \mref {pp:SmoothMu} shows that $S_2^c$ agrees with $I$ for smooth lattice cones. By Corollary~\ref{co:subd}.(\mref {thm:MuSub}), $S_2^c$ has \ISubP, which is known to hold for $I$ also. Therefore, by taking smooth subdivisions, they agree for all lattice cones, proving the first equation. Then the second and third equations follow from Corollary~\mref{cor:abfd} and Proposition~\mref{prop:mu} respectively.
\end{proof}

In view  of Proposition \mref{pp:oeml}, the \abf for $S^o$ yields an open variant of the Euler-Maclaurin formula by applying the factorization to the linear map $S^o:\QQ \cc\to
\calm_\QQ (\C ^\infty)$.

\begin{coro}
{\bf (Open Euler-Maclaurin formula)}  Let $\lC\in \cc_k$ be a lattice cone and $\vec \e \in V_k^*$. Then
$$S^o_2\lC (\vec \e)=I\lC (\vec \e),$$
$$\hspace{-3cm}(\text{\bf Euler-Maclaurin formula})\qquad S^o\lC =\sum _{F\preceq C} \mu^o(t(C,F), \Lambda _{t(C,F)})I(F, \Lambda _F)
$$
and we have a projection formula for the interpolate function $\mu^o$,
$$\mu ^o =\pi _+S^o.$$
\mlabel{coro:cumfo}
\end{coro}

\noindent
{\bf Acknowledgements}:
This work is supported by the National Natural Science Foundation of China (Grant No. 11071176, 11221101 and 11371178) and the National Science Foundation of US (Grant No. DMS~1001855). The authors thank Kavli Institute for Theoretical Physics China (KITPC) and Morningside Center of Mathematics (MCM) in Beijing where part of the work was carried out.
The second author thanks Sichuan University, Lanzhou University and Capital Normal University for their kind hospitality.

\end{document}